\documentclass[11pt]{article}
\usepackage[margin=1in,a4paper]{geometry}

\newcommand{\TitleMetadataForHyperref}{Robust Shattering Arguments}
\newcommand{\AuthorMetadataForHyperref}{Mohsen Ghaffari, Magnús M. Halldórsson, Yannic Maus, Alexandre Nolin}

\usepackage{1a_generic_macros}

\usepackage{1b_paper_macros}

\usepackage{tikz}

\begin{document}

\title{Robust Shattering Arguments}

\author{
Mohsen Ghaffari {\small\orcidlink{0000-0003-4213-9898}},
Magnús M.\ Halldórsson {\small\orcidlink{0000-0002-5774-8437}},
Yannic Maus {\small\orcidlink{0000-0003-4062-6991}},
and Alexandre Nolin {\small\orcidlink{0000-0002-3952-0586}}
}

\date{}
\maketitle

\begin{abstract}
Graph shattering is a central technique underlying sublogarithmic-time distributed algorithms in the $\mathsf{LOCAL}$ model. Its analysis typically relies on bounding the probability that large sets of distant nodes remain unresolved, often via independence assumptions justified by locality. 

We show that these assumptions fail for pre-shattering procedures that run for super-constant rounds, where dependencies accumulate over time. As a result, several standard shattering arguments in the literature are incomplete, including those for maximal independent set, $(\Delta+1)$-coloring, and the distributed Lovász Local Lemma (LLL).

We provide a systematic repair of these analyses. Our main contribution is a corrected shattering analysis of the Fischer--Ghaffari LLL algorithm. In addition, we develop general tools that capture common patterns in modern algorithms and yield the required decay bounds without relying on independence. We also present explicit counterexamples to commonly used shattering lemmas.

Overall, we establish a robust and reusable foundation for shattering arguments in the presence of long-range dependencies.
\end{abstract}

\begin{center}
\large
Due to ample requests we are sharing this preliminary draft though we are currently working on a generalization of the claims in this paper. 
\end{center}

\clearpage
\tableofcontents
\clearpage

\section{Introduction}

Rigorous reasoning is the foundation of theoretical computer science. Yet the field’s conference-driven pace can at times encourage arguments that are intuitive but not fully justified. This tension is particularly pronounced in probabilistic analyses, where subtle dependencies can invalidate otherwise compelling lines of reasoning.

A central example is the technique of \emph{graph shattering}, which underlies many of the fastest randomized graph algorithms across computational models. The shattering paradigm proceeds in two phases: a randomized \emph{pre-shattering phase} that solves most of the instance, followed by a \emph{post-shattering phase} that completes the computation on the remaining subgraph. The key property is that the residual graph decomposes into small connected components, enabling efficient deterministic processing.

Shattering has become indispensable in the design of sublogarithmic-time distributed algorithms in the \LOCAL model~\cite{linial1992locality}. It underlies state-of-the-art results for problems such as maximal independent set (MIS)~\cite{Ghaffari_soda16,Ghaffari19}, $(\Delta+1)$-coloring~\cite{CLP20,HKNT_stoc22}, maximal matching~\cite{BEPS16_jacm}, $\Delta$-coloring~\cite{GHKM18}, sinkless orientation~\cite{GS_soda17}, and the constructive Lovász Local Lemma (LLL)~\cite{FG17,GHK18,CHLPU20,MU21,Davies_soda23}. Moreover, the randomized complexity of such problems is lower bounded by the deterministic complexity on exponentially smaller instances~\cite{CKP19}, closely matching the structure produced by shattering.

At the heart of any shattering argument lies a probabilistic claim about the residual graph after the pre-shattering phase. Specifically, one must show that
\begin{quote}
(*) every connected component of the remaining graph has size $\poly(\log n)$ with high probability.
\end{quote}

The standard approach is to argue that each node remains unresolved with small probability, and that dependencies are sufficiently local.
Under such locality, events on well-separated nodes behave independently, allowing probabilities to multiply. 
This rules out the existence of large residual components as they necessarily contain many such nodes.

The crucial assumption is that dependencies are local. While this holds for constant-round algorithms, many of the most powerful shattering-based results rely on pre-shattering procedures that run for super-constant rounds. In this regime, dependencies accumulate over time, and the standard independence arguments no longer apply.

In this work, we show that a number of existing shattering proofs rely on independence assumptions that do not hold in general. As a consequence, several widely cited arguments are incomplete or incorrect. Nevertheless, we demonstrate that the underlying algorithmic results are largely correct, and we provide new analyses that restore them on rigorous foundations.

While parts of the correct reasoning are known informally, they are neither systematically documented nor consistently reflected in the literature. Incorrect arguments have been repeatedly reused across central results. Our goal is to provide a unified and fully rigorous foundation for shattering arguments.

\subsection{Our Contribution}

Our contributions are primarily foundational: we identify a systematic issue in a widely used analytical technique, demonstrate concrete failure modes, and develop tools that restore correctness while preserving the power of the shattering framework.

A recurring assumption in the analysis of shattering-based algorithms---that events at sufficiently distant nodes can be \emph{easily} treated as if independent---is invalid in general once the pre-shattering phase spans $\omega(1)$ rounds.
While independence follows from locality in constant-round algorithms, it breaks down when dependencies accumulate over time.
We demonstrate that this issue is not merely technical: it affects several widely cited results, whose analyses rely on independence or related ``local conditioning'' arguments that do not hold.

\paragraph{Explicit counterexamples.}
To make this failure concrete, we construct simple counterexamples showing that both common proof strategies for establishing shattering can fail:
(i) arguments that rely on independence of distant events, and 
(ii) arguments that attempt to replace independence by robustness under adversarial choices of randomness outside a local neighborhood.
These examples demonstrate that neither approach suffices to establish the key separated-set bound required for shattering.

\paragraph{A corrected analysis of the distributed LLL (main technical contribution).}
Our main technical contribution is a new analysis of the Fischer--Ghaffari algorithm for the distributed Lovász Local Lemma.
We identify a flaw in the original shattering argument—specifically, an incorrect independence claim—and replace it with a rigorous analysis that avoids independence entirely.
Our approach models the pre-shattering phase as a sequential sampling process and relates partial executions to fully independent assignments.
This yields the required $\Delta^{-\Omega(|S|)}$ bound for well-separated sets and restores the correctness of the algorithm’s complexity guarantees.
As a consequence, previously claimed separations in the distributed time hierarchy are reinstated.

\paragraph{General tools for handling dependencies over time.}
Beyond the LLL setting, we develop general techniques for establishing shattering in the presence of long-range dependencies.
Our approach isolates the core requirement needed for shattering—a strong decay bound on the probability that a separated set survives—and provides reusable ways to obtain it without relying on independence.
These tools apply to several classes of algorithms where dependencies evolve over multiple rounds, including symmetry-breaking procedures such as MIS, maximal matching, and graph coloring.

\paragraph{A unified and rigorous framework.}
Taken together, our results provide a clean and robust foundation for the graph shattering paradigm.
We clarify the precise conditions under which shattering holds, separate combinatorial and probabilistic components of the argument, and eliminate recurring sources of error.
This yields a reusable framework that can be applied to existing algorithms and guides the design and analysis of future ones.

\subsection{Content of the paper}
In \Cref{sec:tecOverview}, we provide a technical overview of shattering arguments, highlighting key challenges and common pitfalls, and outlining our approach. In \Cref{sec:counterExample}, we present counterexamples to commonly used assumptions in prior work. In \Cref{sec:shattering}, we isolate general and robust shattering tools used throughout the paper. As a side benefit, it yields somewhat better parameters for shattering than before.
In \Cref{sec:fg-lll}, we revisit the Fischer--Ghaffari LLL algorithm and provide a corrected analysis. In \Cref{sec:multiPhase,sec:opportunity}, we extend our framework to more general settings, including opportunity-based and multi-phase processes that can be used to re-establish shattering results for several fundamental distributed algorithms; in particular, in \Cref{sec:multiPhase} we prove shattering for the state of the art $\Delta+1$-vertex coloring algorithm.

In \Cref{sec:ballCover}, we present a small ball cover argument for post-shattering components. This is particularly useful in regimes where the component size bound of $\poly(\Delta)\log n$ is not sufficiently small compared to $n$, e.g., when $\Delta \gg \poly \log n$. The small ball cover is a general method to further decompose such components into manageable regions, enabling efficient post-shattering algorithms.

In \Cref{app:OpportunitiesApplications}, we provide applications of our general framework to opportunity-based processes. In particular, we use these tools to provide a self-contained proof of \Cref{thm:MIS} for the state of the art algorithm to compute maximal independent sets.

\section{Technical Overview}
\label{sec:tecOverview}

\subsection{Structure of Shattering Arguments}

We begin by isolating the structure of standard shattering arguments, identifying the step where difficulties arise, and outlining how our work addresses them.

At a high level, existing proofs that a randomized pre-shattering phase leaves only small connected components consist of two largely independent ingredients.

\newcommand{\argOneText}{(A)}
\newcommand{\argTwoText}{(B)}
\newcommand{\argOneLink}{\hyperlink{arg1}{\argOneText}}
\newcommand{\argTwoLink}{\hyperlink{arg2}{\argTwoText}}

\paragraph{Shattering Template.}
Our first step is to extract a specific form of these two parts.
Let $B \subseteq V$ denote the random set of nodes that remain after the pre-shattering phase, i.e., the nodes that participate in the post-shattering instance. We need two terms that remain only partially formalized for now. 
The term "well-separated" means that for a problem specific constant $c$, the nodes of $S$ are of mutual distance at least $c$. Further, $S$ is \emph{structured} if it induces a connected subgraph in $G^{2c}$.
A shattering proof consists of:
\begin{enumerate}[label=(\Alph*)]
\item \hypertarget{arg1}{\textbf{Separated-Set Bound.}} For every well-separated set $S$ of pairwise distant nodes,
\begin{equation}
\label{condition:THESETsimplified}
\Pr[S \subseteq B] \le \Delta^{-\Omega(|S|)}.
\end{equation}
\item \hypertarget{arg2}{\textbf{Witness Extraction + Union Bound.}} 
Every large connected component in the post-shattering graph contains a structured subset $S$ of distant nodes. 
Enumerating such sets and applying a union bound shows that large components are unlikely.
\end{enumerate}

The second ingredient is combinatorial and essentially independent of the algorithm; we provide it in \Cref{sec:Generalshattering}. All algorithm-specific difficulty lies in establishing \argOneLink, and that's where the central technical task in any shattering argument lies.

\subsection{How Is \argOneText\ Established?}
\label{sec:arg1}

We first illustrate a setting where \argOneLink\ is straightforward. Consider the problem of coloring a graph of maximum degree $\Delta$ with $\Delta^{10}$ colors. Each node independently selects a color uniformly at random and keeps it if no neighbor selects the same color; otherwise, it remains uncolored.

A node remains uncolored with probability $1/\Delta^9$, and these events are independent for nodes at distance greater than $2$. Hence, for any $2$-independent set $S$, the probability that all nodes in $S$ remain uncolored is $\Delta^{-9|S|}$, establishing \Cref{condition:THESETsimplified}. By 
the lemma establishing \argTwoLink, the remaining graph consists of small connected components w.h.p.

This argument relies on a key property of local algorithms: if a \LOCAL algorithm runs for $r$ rounds, the outcome at a node depends only on randomness within its $r$-hop neighborhood. Consequently, events at nodes at distance greater than $2r$ are independent. This locality makes \argOneLink\ easy to establish for constant-round processes.

The situation changes fundamentally when the pre-shattering phase runs for $\omega(1)$ rounds. In this regime, dependencies can propagate over time, and distant nodes may become correlated even if they are far apart in the graph. As a result, the standard independence-based argument for \Cref{condition:THESETsimplified} no longer applies.

A common approach in the literature is to argue that each node is solved with high probability even under adversarial choices of randomness outside a bounded neighborhood. While this robustness may suggest that distant nodes behave independently, this inference is incorrect: such guarantees neither imply independence nor suffice to establish \argOneLink. We give explicit counterexamples in \Cref{sec:counterExample}.

\subsection{The Primary Application: The Fischer--Ghaffari LLL Algorithm}

We now turn to the Fischer--Ghaffari algorithm for the distributed Lovász Local Lemma (LLL), which serves as our central example where the standard approach to establishing \argOneLink\ fails. It also requires the most non-trivial solution.

Given a set of random variables and a family of ``bad'' events, each depending on a small subset of the variables, the goal of the constructive LLL (see \cref{sec:LLLdef}) is to compute an assignment to the variables that avoids all bad events, under conditions ensuring that such an assignment exists. 
There are two key parameters: $p$ and $d$. The former is an upper bound on the probability of each bad event, under the assumed random process for sampling the variables.
Also, each event is independent of all but at most $d$ other events. The objective is to compute an assignment to the random variables that avoids all the bad events.

Erd\H{o}s and \lovasz \cite{EL74} showed that the LLL goal can be achieved whenever $4pd \le 1$.
After a series of contributions on constructive approaches, Moser and Tardos \cite{MoserTardos10} achieved this with an efficient algorithm. Their approach leads to a $O(\log n)$-round distributed algorithm, under an improvement of Chung, Pettie and Su \cite{CPS17}. It is known that $\Omega(\log\log n)$ rounds are necessary \cite{BGR21}, while research is ongoing on the question if the gap can be closed; in fact, it is conjectured that the complexity is $O(\log\log n)$ for constant $d$ under criterion $pd^c<1$ for some constant $c$~\cite{CP19}. In the case of low-degree graphs, a sublogarithmic-time algorithm was presented by Fischer and Ghaffari \cite{FG17} for a weaker but highly important polynomial condition.

In the Fischer--Ghaffari algorithm, variables are progressively assigned values while some variables are marked as \emph{frozen}. An event becomes \emph{dangerous} when its marginal probability exceeds a threshold, at which point its remaining variables are frozen. The post-shattering instance consists of these frozen variables and their incident events.

The process proceeds sequentially over a coloring of the event graph, so that variables are exposed gradually. This sequential exposure is crucial for the algorithm's correctness but comes with a challenge: the state of an event when it is processed depends on earlier random choices.

The original analysis relies on the following claim.

\begin{falseObservation}[Observation 10 in~\cite{FG17}]
For each event $A \in X$, the probability that $A$ has at least one unset variable is at most $(d + 1)\sqrt{p}$. Furthermore, this is independent of events at distance greater than $2$ from $A$.
\end{falseObservation}

The probability bound is correct, and the independence claim is natural by analogy with constant-round algorithms. However, it fails in this setting. Because the process unfolds over many rounds of sequential sampling, the state of an event depends on earlier sampling decisions. When an event becomes dangerous, this can influence the state of other events, and such effects can propagate over many rounds.

Thus, participation in the post-shattering phase is not independent for distant events. In particular, this invalidates the independence assumption needed to establish \argOneLink. In \Cref{app:examples}, we give a concrete instance demonstrating this phenomenon.

We re-establish the following result.
Let $\Det_{\mathsf{LLL},d,q}(N)$ denote the deterministic complexity of LLL on instances of size $N$, dependency degree $d$, and maximum event probability $q$.

\newcounter{savedLLL}
\setcounter{savedLLL}{\thetheorem}

\addtocounter{savedLLL}{1}
\edef\temp{theorem.\thesavedLLL.simple}
\addtocounter{savedLLL}{-1}

\hypertarget{\temp}{
 \begin{restatable}[\cite{FG17} redux]{theorem}{thmLLLSimple}
\customlabel[theorem]{thm:mainLLL}{simple}
There is a randomized distributed algorithm for the constructive \lovasz Local Lemma under criterion $8pd^{9} < 1$ that w.h.p.\ runs in 
$O(d^2) + \Det_{\mathsf{LLL}, d, \frac{1}{3(d+1)}}(\poly\log n)$ rounds.
\end{restatable}
}

This theorem has several immediate implications.
It yields an algorithm with comparable guarantees in the \CONGEST model \cite{MU21}, and it reinstates correctness for the $(1+\eps)\Delta$ edge-coloring algorithm of \cite{HMN_disc22}, which uses the LLL routine as a subroutine. Most importantly, it repairs a gap in the widely studied complexity landscape of locally checkable labeling (LCL) problems, e.g., ~\cite{CKP19,CP19,BHKLOS18,BBOS18,RG20}. For constant-degree graphs, any LCL solvable in $o(\log n)$ rounds can be reduced in constant time to an LLL instance with $d=O(1)$ \cite{CP19}. Consequently, Theorem~\ref{thm:mainLLL} re-establishes the separation between $\tilde{O}(\log^3 \log n)$ and $o(\log n)$ in the distributed time hierarchy. 

\subsection{Our Approach (Informal)}

The key observation is that for monotone sequential sampling processes, the probability of reaching a given partial assignment depends only on the revealed variable values and not on the order in which they are exposed. This allows us to reinterpret the process as if remaining variables were sampled independently at the end.

Using this perspective, we relate the probability of reaching a configuration in which a set $S$ remains unresolved to the probability that, under a full independent assignment, certain bad events occur simultaneously. This yields the required bound {\argOneLink}
without assuming independence during the process.

\subsection{Extensions: Beyond the LLL Setting}
\label{sec:extensions}

While our main focus is the LLL setting, the same difficulty in establishing \argOneLink\ arises in other classes of randomized algorithms. In these settings, dependencies accumulate over time, preventing direct independence-based arguments.

We re-establish the following seminal results involving randomized distributed algorithms.
Recall that $\deg+1$-list coloring is the problem of computing a proper vertex coloring when each vertex $v$ of degree $d(v)$ chooses from a given set of $d(v)+1$ colors. 
Let $\Det_{\mathsf{d1LC}}(N)$ be the deterministic complexity of that problem on a graph on $N$ vertices.
Similarly, let $\Det_{\mathsf{MM}}(N)$ denote the same quantity for the Maximal Matching problem.

\begin{itemize}
\item Maximal Matching 
    in $O(\log \Delta) + O(\Det_{\mathsf{MM}}(\poly \log n))$ rounds~\cite{BEPS16_jacm};

    \item Maximal Independent Set in $O(\log \Delta) + \Otilde(\log^3 \log n)$ rounds~\cite{Ghaffari_soda16};

     \item $\Delta+1$-list coloring in $O(\Det_{\mathsf{d1LC}}(\polylog n))$ rounds~\cite{CLP20}.
\end{itemize}

These results follow from our treatment of more general phenomena.
We specifically analyze two common patterns.

\paragraph{Opportunity-based decision processes.}
A particular pattern arises in algorithms such as Ghaffari’s MIS algorithm~\cite{Ghaffari_soda16} and the maximal matching algorithm of~\cite{BEPS16_jacm}, where nodes are not guaranteed to be solved in any single round with high probability.

Instead, nodes accumulate progress over time through a sequence of probabilistic opportunities. While opportunities in a given round are typically independent for well-separated nodes, the timing of these opportunities may be highly correlated across nodes. Despite this, we show that one can still obtain strong bounds on the probability that many nodes fail to make progress, thereby establishing \argOneLink.

\newcounter{savedOpportunityProcess}
\setcounter{savedOpportunityProcess}{\thetheorem}

\addtocounter{savedOpportunityProcess}{1}
\edef\temp{theorem.\thesavedOpportunityProcess.formal}
\addtocounter{savedOpportunityProcess}{-1}
\hypertarget{\temp}{
\begin{restatable}{theorem}{thmOpportunities}
\label{thm:probability-guarantee}
\customlabel{thm:opportunityProcess}{formal}
Consider an progressive opportunity process (\cref{def:opportunistic-process}) where a set of nodes $S$ is guaranteed to experience $x$ opportunities to be solved of failure probability $p$ by the end of the process. Then the probability that all nodes in $S$ made no progress at the end of the process is at most $p^x$.
\end{restatable}
}
In \Cref{thm:opportunityProcessGeneralized} (\Cref{ssec:generalizedOpportunityProcess}) we present a generalization of 
\Cref{thm:probability-guarantee}.
\paragraph{Multi-phase residual processes.}
Many modern randomized graph algorithms, such as the coloring algorithms of~\cite{CLP20,HKNT_stoc22}, proceed in multiple phases. In each phase, nodes attempt to make progress based on local conditions, and nodes that fail are deferred to a residual graph handled at the end. 

While each phase can often be analyzed in isolation, dependencies accumulate across phases, making it difficult to bound the probability that a set of nodes survives all phases. We show that it nevertheless suffices to control the failure probability of sets of nodes within each phase, even under arbitrary conditioning on previous phases, yielding the required bound for \argOneLink.

\newcounter{savedMultiPhase}
\setcounter{savedMultiPhase}{\thetheorem}

\addtocounter{savedMultiPhase}{1}
\edef\temp{theorem.\thesavedMultiPhase.informal}
\addtocounter{savedMultiPhase}{-1}

\hypertarget{\temp}{
 \begin{restatable}[Informal]{theorem}{thmMultiPhaseSimple}
\customlabel{thm:multiPhaseSimple}{informal}
 Any multi-phase residual process leads, with high probability, to small unsolved components. 
 \end{restatable}
}
\medskip

These abstractions provide general and reusable tools for establishing shattering in settings where independence-based arguments fail. 
We develop the former abstraction in \cref{sec:opportunity}
and the latter in \cref{sec:multiPhase}.

\subsection{Related work}
The body of work using shattering is far too large to revisit in full, and doing so is beyond the scope of this paper. Shattering traces its roots back to Beck’s algorithmic treatment of the Lovász Local Lemma, where he showed that after a randomized resampling phase, the remaining “bad” events decompose into small, independent components that can be fixed deterministically \cite{Beck_rsa91a}. It was later used in various LLL-type problems like hypergraph colorings in  centralized and parallel settings, e.g.,  \cite{MR_stoc98,Alon91,PT09}.  Since its introduction into the distributed setting by Barenboim, Elkin, Pettie, and Schneider in \cite{BEPS16_jacm} and into the distributed-adjacent setting of local centralized computation by Rubinfeld, Tamir, Vardi, and Xie in~\cite{RTVX11},  it has become a ubiquitous technique across various computational models, including the LOCAL model \cite{BEPS16_jacm,Ghaffari_soda16,GHKM18,CHLPU20,Davies_soda23,HMN_disc22,GS_soda17,FG17,CLP20,HKNT_stoc22,BMU25,GHK18,BMNSU25,BGKMU19,AHN_sirocco23}, the CONGEST model \cite{MU21,Ghaffari19,MPU25,HN_tcs23,HKMN20,FHN_disc23,FHN_podc25,FHN_disc24}, the local centralized computation model \cite{RTVX11,BGR21,G22}, and the massively parallel computation model \cite{componentstable,GU19,CDP21}.

Direct methods for LLL -- the seminal Moser-Tardos algorithm  \cite{MoserTardos10} and the Chung-Pettie-Su \cite{CPS17} algorithm -- do not rely on the shattering technique and imply (poly-)logarithmic time distributed algorithms for the constructive \lovasz Local Lemma~\cite{MoserTardos10,CPS17}. For the extensive work on distributed graph coloring, we refer to \cite{barenboimelkin_book} for older works and \cite{HKNT_stoc22} for recent works.

The literature on algorithmic applications of the \Lovasz\ Local Lemma includes a few instructive subtleties. For example, in the parallel algorithm of \cite{MR_stoc98}, Pach and Tardos \cite{PT09} later observed an issue in the proof arising from the use of retractions when a marginal event becomes 'dangerous(ly)' close to hold. Such retractions inadvertently influence other events' probabilistic claims. Also, the flawed independence assumption in the shattering argument of Fischer and Ghaffari \cite{FG17}, whose algorithm is inspired by \cite{MR_stoc98}, may be traced back to a related statement already present in \cite{MR_stoc98}.  Our analysis can also  be used to re-establish their result.

\subsection{Notation}

A graph $G=(V,E)$ is defined by a set of vertices $V$ and a set of edges $E \subseteq \binom{V}{2}$ representing adjacency between those vertices.
The neighborhood of a vertex $v$ in $G$ is denoted $N_G(v)$ and defined as $N_G(v) = \set{u \in V: uv \in E}$.
A path $P$ in a graph $G=(V,E)$ between two vertices $u,v \in V$ is a finite sequence of nodes $u_0,u_1,\ldots,u_k$ s.t.\ $u_0 = u$, $u_k = v$, and $\forall i \in \set{0,1,\ldots,k-1}$, $u_i u_{i+1} \in E$. The integer $k$ is the length of the path.
For any two vertices $u,v \in V$, we denote by $\dist_G(u,v)$ their distance in $G$, defined as $\dist_G(u,v) = +\infty$ if no path between $u$ and $v$ exists in $G$, and $\dist_G(u,v) = k$ if the shortest paths between $u$ and $v$ in $G$ has length $k$.
The connected component of a node $u\in V$ in the graph $G$ is the set $\CC_G(u) = \set{v \in V \mid \dist_G(u,v) < +\infty}$.
When the graph $G$ is clear from context, we remove the subscript $G$ and write $N(v)$ and $\dist(u,v)$.

For an integer $x \geq 1$, the $x$-th power graph of $G$, $G^x = (V,E^x)$, is the graph defined on the sames vertices $V$ as $G$ with an edge $uv\in E^x$ iff $\dist_G(u,v) \leq x$ for all $u,v \in V$, $u\neq v$.
For a subset of the vertices $S \subseteq V$ of a graph $G=(V,E)$, $G[S]$ is the subgraph of $G$ induced by $S$, of vertex set $S$ and edge set $E[S] = \set{uv\in E \mid u\in S \wedge v \in S}$.

For any integer $k \geq 0$ we denote by $[k]$ the set $\set{i \in \integers : i \leq k}$ -- in particular, $[0] = \emptyset$.
For functions $f,g: \integers \to \reals$ we denote by $f(n) \in \Otilde(g(n))$ that there exists a constant $c$ s.t.\ $f(n) \in O(g(n) \log^cg(n))$. Asymptotics are always taken as $n \to +\infty$ unless stated otherwise.

\section{Failure of Standard Shattering Arguments}
\label{sec:counter-examples}

The technical overview identified two natural approaches for establishing the Separated-Set Bound \argOneLink. We now show that both approaches fail in general, even though they are widely used in the literature.

\subsection{Failure of Independence-Based Arguments}
\label{sec:FGShattering}

The following lemma is the basic shattering statement for genuinely local pre-shattering procedures. When a vertex remains unresolved based only on bounded-radius randomness, sufficiently distant vertices behave independently, and the usual witness-set union bound gives small residual components with high probability.
\begin{lemma}[\cite{FG17}]\label{lem:Shattering}
Let $G=(V, E)$ be a graph with maximum degree $\Delta$. Consider a process that generates a random subset $B \subseteq V$ such that $\Pr[v \in B]\leq \Delta^{-c_1}$, for some constant $c_1 \geq 1$, and such that the random variables $1(v\in B)$ depend only on the randomness of nodes within at most $c_2$ hops from $v$, for all $v\in V$, for some constant $c_2\geq 1$.
Then, for any constant $c_3\geq 1$, satisfying  $c_1>c_3+ 4c_2 + 2$,  we have that any connected component in $G[B]$ has size at most $O( \log_{\Delta} n  \Delta^{2c_2})$ with probability at least $1- n^{-c_3}$.
\end{lemma}
This lemma is correct, but many algorithms  do not meet its requirement. In particular the algorithm in the paper that presented the lemma does not meet it, see below. This lemma is helpful whenever the runtime of the pre-shattering phase is limited to $r=O(1)$ rounds as that immediately implies that the event $1(v\in B)$ only depends on the randomness in its $r$-hop neighborhood. In other words, the above lemma can be applied with $c_2=r$. Examples of such cases are the sinkless (hypergraph) orientation algorithms in \cite{GS_soda17,BMNSU25}, the $\Delta$-coloring algorithm from \cite{GHKM18}, the vertex and degree splitting algorithms 
\cite{BGKMU19,HMN_disc22}, and also the LLL algorithms from  \cite{GHK18,HMP24}  that work for special LLL instances. 
In a nutshell, it is difficult or even impossible to apply the lemma when the pre-shattering phase uses $\omega(1)$ rounds, as is the case in the LLL paper that introduced the lemma \cite{FG17}\footnote{See the discussion in \Cref{sec:tecOverview} and \Cref{app:examples} for concrete counterexamples establishing long-range dependencies on distant randomness.}, but also most other algorithms for symmetry-breaking such as the maximal independent set, maximal matching, or graph coloring problems or LLL, e.g., \cite{BEPS16_jacm,Ghaffari_soda16,GHK18,Davies_soda23,CLP20,CLP18,HKMT21,HKNT_stoc22,HNT22}.

\subsection{Adversarial-based Shattering Argument}
\label{sec:counterExample}
One way to address potential long-range dependencies is to show that the process under consideration is stochastically dominated by another process in which the participation of distant nodes is genuinely independent, as done, for example, in \cite{RTVX11}. However, in establishing \argOneLink, one must carefully distinguish between stochastic domination of the entire process and stochastic domination within a single iteration.

The following statement captures an intuitively appealing line of reasoning for establishing shattering: namely, that participation in the post-shattering phase should remain unlikely for any given node even if the randomness outside its local neighborhood is fixed adversarially. As we shall see, however, this intuition is not formally correct in general.

\begin{semma}[Lemma 4.1 in \cite{CLP20}]
\label{lem:shatteringWrong}
Consider a randomized procedure that generates a subset of vertices $B \subseteq  V$. Suppose that for each $v \in  V$, we have $\Pr[v \in  B] \leq \Delta^{-3c}$, and this holds even if the random bits not in $N^c(v)$ are determined adversarially. With probability at least $1-n^{-\Omega(c^\prime)}$, each connected component in the graph induced by $B$ has size at most $(c^\prime / c)\Delta^{2c} \log_\Delta  n$.
\end{semma}
This lemma also appears and hence is used wrongly in \cite[Theorem 5.1]{GHK18} though the respective algorithm can be fixed easily. 

The word \emph{adversarial} in this statement has at least two natural definitions.
In a weak form  of adversariality, the adversary must make its choices for the random bits outside $N^c(v)$ \emph{before} the random bits inside $N^c(v)$ are chosen, and in a strong form, the adversary can choose the random bits outside $N^c(v)$ \emph{after} the random bits inside $N^c(v)$ are chosen.
These correspond to whether the adversary makes its choices with or without knowledge of the random bits in $N^c(v)$. 
The authors clarify that their version of adversariality is the \emph{weak form} in the paragraph following the statement of their lemma:

\begin{quote}
    Lemma 4.1 obviously applies to randomized procedures that take c rounds. It also applies to $\omega(1)$-round procedures that are composed of a \emph{series} of $c$-round experiments, where vertices that fail to satisfy some invariant are included in $B$ immediately after the experiment. What is important is that the bound $\Pr[v \in  B] \leq  \Delta^{-3c}$ holds if an adversary is allowed to completely control how the series of experiments proceeds outside $N^c (v)$, so long as it cannot see the random bits generated inside $N^c (v)$.
\end{quote}

For the proof of this lemma, \cite{CLP20} refers to Chang's thesis \cite{changthesis}.
The problem with the proof is that it assumes that the events $[v \in B]$ are independent, something that appears at first sight to be plausible given the hypothesis. 
We outline here two counterexamples to this intuition.

\paragraph{First counterexample to \cref{lem:shatteringWrong}}
Consider the following random process:
\begin{itemize}
    \item Each $v \in V$ picks a random number $x_v \in [\Delta^a]$
    \item Each $v$ is added to $B$ iff there exists a $u \in N^b(v), w \in N^c(u)$ such that $x_u = x_w$.
\end{itemize}

The probability that a node $u$ has the same value $x_u$ as one of its $c$-hop neighbors is at most $\card{N^c(u)} / \Delta^a \leq \Delta^{c-a}$. This is true even if the random bits of all the nodes except $u$ are chosen adversarially in the weak sense (but not in the strong sense). And the probability that a node $v$ gets added to $B$ is bounded by $\Delta^{b+c-a}$, and this is true even if random bits outside the $N^b(v)$ are chosen weakly adversarially.

With $a=6, b=1, c=2$, each node gets added to $B$ with probability at most $p=\Delta^{b+c-a} = \Delta^{-3}$, even with adversarial random choices outside $N^b(v) = N^1(v)$. In every proof of every shattering lemma, nodes at distance $2b+1=3$ from each other would then be considered. For two nodes, we would argue that they both get added to $B$ with probability at most $p^{2} = \Delta^{-6}$, and for a set of $t$ nodes that are pairwise $3$-distant ($2$-independent), we would argue that they all get added to $B$ with probability at most $p^{t} = \Delta^{-3t}$.

But this is not true here because the events are not independent.
Consider the two extremal nodes in the following picture:

\begin{figure}[ht]
    \centering
    \includegraphics{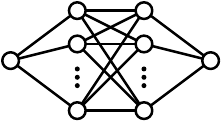}
    \caption{Example of a problematic case, a complete bipartite graph with two additional extremal nodes.}
    \label{fig:shattering_counter_example}
\end{figure}

The graph has four layers numbered from $1$ to $4$, with one node in layers $1$ and $4$ and $\Delta-1$ nodes in layers $2$ and $3$. The probability that there exists a node in layer $2$ with a number equal to a number in layer $3$ is of order $\Delta^{-4}$, since each of the $\Delta-1 = \Theta(\Delta)$ nodes in layer $1$ has a $\Theta(\Delta / \Delta^{6}) = \Theta(\Delta^{-5})$ chance each of hitting one of the up to $\Delta -1$ numbers present in layer $3$. Therefore, the probability that the $t=2$ extremal nodes get both added to $B$ is of order $\Delta^{-4} \gg \Delta^{-3 \cdot t} = \Delta^{-6}$. This shows that we cannot simply multiply the probabilities as if the events $[v \in B]$ were independent: weak-adversariality is not sufficient to claim that.

\paragraph{Second countereample to \cref{lem:shatteringWrong}}

We give an even more striking example in this section, which relies on dependence at large distances.
Consider the following process:
\begin{itemize}
    \item Every node $v \in G$ picks a random number $x_v \in [\Delta^c]$, and
    \item $v$ joins the set $B$ iff 
      $\left(\sum_{u \in G} x_u\right) \mod \Delta^c = 0$.
\end{itemize}

In this case, each node has a probability at most $\Delta^{-c}$ to join $B$, even conditioned on arbitrary random choices by other nodes in the graph (weak-adversariality). But with probability $\Delta^{-c}$, the whole graph joins $B$, which means that the process clearly does not shatter.

\section{A Correct Shattering Condition}
\label{sec:shattering}
\label{sec:Generalshattering}

Having established that standard approaches fail, we now identify a condition that is both sufficient and robust, 
formalizing the Separated-Set Bound \argOneLink{} of 
\cref{sec:tecOverview}.
This condition avoids independence assumptions and will serve as the main tool in our analysis of the LLL algorithm.

\begin{definition}[$a$-independent, $b$-connected]
    \label{def:indep-connect-set}
    A set $S$ of nodes in some graph $G$ is said to be $a$-independent if $G^a[S]$ is an independent set, and $b$-connected if $G^b[S]$ is a connected graph.
\end{definition}
Note that nodes in an $a$-independent set have pairwise distance at least $a+1$.

In order to perform an efficient union-bound in the general shattering proof we use the following bound on the number of $x$-connected sets. Previous proofs of the shattering lemma simply upper bounded this number by $4^t\cdot \Delta^{x(t-1)}$ by using the $2t$-length encoding of unlabeled trees via their Euler tour. However, that bound is not tight enough to prove the shattering lemma for $\Delta=3$. 
\begin{lemma}
\label{lem:x-connected-set-bound}
Let $\Delta\geq 2$ and $t,x\geq 0$ be integers. An $n$-node graph with maximum degree $\Delta$ has at most $e^t\cdot \Delta^{x(t-1)}$ $x$-connected sets of size $t$.  $\Delta^{x}$ can be replaced with $\max_{v\in V}|N^x(v)|$.
\end{lemma}
\begin{proof}
The number of $x$-connected unlabeled sets of size $t$ is a $t!$ fraction of the number of labeled $x$-connected sets of size $t$. 
We can associate with each labeled set a labeled spanning tree of size $t$, together with an embedding of the spanning tree in the graph such that each combination of spanning tree and embedding is associated with at most one labeled set. The embedding and the spanning tree are formed by greedily iterating through the nodes of the set such that the next node is within distance $x$ from the already picked vertices. 
As there are at most $t^{t-2}$ labeled spanning trees on $t$ vertices by Cayley's formula \cite{aigner2018proofs}, and there are $n\cdot \Delta^{x(t-1)}$ ways to embed the tree into the graph, $n$ choices for the vertex with the lowest label and $\Delta^x$ choices for each remaining vertex. Thus, the number of $x$-connected unlabeled sets is at most 

\begin{align}
\frac{1}{t!}t^{t-2} n\Delta^{x(t-1)}
\leq \frac{t^{t-2}}{e \cdot (t/e)^{t}} n\Delta^{x(t-1)}
\leq e^t\cdot \Delta^{x(t-1)}~, 
\end{align}
where we used $n!\geq e\cdot (n/e)^n$.
In this argument, $\Delta^x$ can be replaced with $\max_{v\in V}|N_x(v)|$, i.e., the maximum size of an $x$-hop neighborhood of a vertex in the graph. 
\end{proof}

\begin{lemma}[Shattering Lemma]
\label{lem:ShatteringGeneral}
Let $G=(V, E)$ be an $n$-node graph with maximum degree $\Delta\geq 2$, and let $c_2,c_3,x\geq 0$ . Consider a process that generates a random subset $B \subseteq V$ such that 
\begin{align} \label{condition:THESET}
    P[S\subseteq B]\leq \Delta^{-(x+c_3+2) |S|}
\end{align}
holds, for each $c_2$-independent set $S\subseteq V$ of size $\log_\Delta n$.

Then with probability at least $1- n^{-c_3}$, every $c_2$-independent $x$-connected (distances measured in $G$) that is fully contained in $B$ has size at most $\log_\Delta n$. More generally, for every $t\geq \log_{\Delta} n$, if \cref{condition:THESET} holds for every set of size $t$, with probability at least $1- \Delta^{-c_3\cdot t}$, every $c_2$-independent $x$-connected set that is fully contained in $B$ has size at most $t$ . 
\end{lemma}
\begin{proof}
We first prove the result for $\Delta\geq 3$.
We do a union bound over all the $x$-connected sets of size $t \geq \log_\Delta n$ contained in the graph $G$.
By \cref{lem:x-connected-set-bound}, there are at most $e^t \cdot \Delta^{x(t-1)}$ such sets. As each such set $S$ has a probability of at most $\Delta^{-x+c_3+2}$ of being contained in $B$, the probability that one of them is fully contained in $B$ is at most

\begin{align*}
  e^t\cdot n \Delta^{x(t-1)}\cdot \Delta^{-(x+c_3+2)t}
  & =
  e^t\cdot n \Delta^{-x-2t} \cdot \Delta^{-c_3 t} \\
  & =
  (e^t \cdot \Delta^{-t}) \cdot (n \cdot \Delta^{-t}) \cdot \Delta^{-c_3 t} \\
& <
  (n \cdot \Delta^{-t}) \cdot \Delta^{-c_3 t} & (\Delta\geq 3 > e)\\
  & \leq
\Delta^{-c_3 t}. & (t \geq \log_\Delta n)
\end{align*}
The case $\Delta = 2$ follows similarly, remarking that $G$ is made of paths and cycles. This bounds the distance-$x$ neighborhood of each node in the graph by $2x-1$, and the number of $x$-connected sets of $t$ nodes by $n\cdot (2x)^{t-1}$ (\cref{lem:x-connected-set-bound}).
\end{proof}

Note that \Cref{lem:ShatteringGeneral} only makes sense for $x>c_2$ as there are no $c_2$-independent sets that are $x$-connected when $x\leq c_2$.

\begin{corollary}[Post-shattering connected components are small]\label{cor:shattering-small-cc} 
Let $G=(V, E)$ be an $n$-node graph with maximum degree $\Delta\geq 2$, and let $c_2,c_3\geq 0$ . Consider a process that generates a random subset $B \subseteq V$ such that 
\begin{align} 
    \Pr[S\subseteq B]\leq \Delta^{-(c_2+c_3+3) |S|}
\end{align}
holds, for each $c_2$-independent set $S\subseteq V$ of size $t\geq \log_\Delta n$. Then with probability at least $1-\Delta^{-c_3t}$ every connected component of $G[B]$ has size at most $t \Delta^{c_2}$.
\end{corollary}
\begin{proof}
We show that any connected component $C\subseteq B$ of size at least $\Delta^{c_2}t$ contains a $c_2$-independent $c_2+1$-connected subset $S$ of size $|C|/\Delta^{c_2}$. Given $C$, the set $S$ can be constructed greedily. Let $v$ be an arbitrary node of $C$, begin with $S=\{v\}$ and remove $v$'s $c_2$-hop neighborhood $N^{c_2}(v)$ from $C$. Now, pick another remaining node in $C$ that is $c_2+1$-connected to $S$. If $C$ is not empty such a node has to exist as $C$ is connected. Each step removes at most $\Delta^{c_2}$ nodes from $C$ and hence at the end of the process $S$ contains at least $|C|/\Delta^{c_2}=t$ nodes, is $c_2$-independent and $c_2+1$-connected.
\end{proof}

\section{The Fischer–Ghaffari LLL Algorithm:  Shattering Analysis}
\label{sec:fg-lll}

We now present the main technical contribution of the paper. Namely, we now revisit the Fischer--Ghaffari LLL algorithm and show how to establish the Separated-Set Bound {\argOneLink} for use in the framework of \cref{sec:shattering}. This resolves the flaw identified in \cref{sec:arg1,app:examples}  and restores the correctness of the shattering argument.

Our analysis proceeds in three steps: (i) reinterpret the process as sequential sampling, (ii) bound probabilities of partial assignments, and (iii) relate these to independent executions.

\thmLLLSimple*

\subsection{\Lovasz Local Lemma Definition}
\label{sec:LLLdef}
\paragraph{Lovász Local Lemma.} We start with an LLL instance described by a bipartite graph $B = (V_{\cE},V_{\var},E)$ where the nodes in $V_{\cE}$ represent events, while the nodes in $V_{\var}$ represent independent random variables. For each event node $v \in V_{\cE}$ and variable node $u\in V_{\var}$, an edge between $u$ and $v$ indicates that the event $\cE_v$ is influenced by the variable $X_u$. Thus, for each event $\cE_v$, the neighborhood of $v$ in $B$ is the set of random variables $\var(v) = \var(\cE_v)$ deciding $\cE_v$.

Non-independence between events is further represented by the dependency graph $H=(V_{\cE},E')$, where $uv \in E'$ iff $\var(u) \cap \var(v) \neq \emptyset$. Put differently, $ H = B^2[V_{\cE}]$. Next we define two important parameters of an LLL instance. Let $p$ be an upper bound on the maximum probability of any event in the LLL, i.e.,  $\max_{v\in V} \Pr[\cE_v]\leq p$, and $d$  be an upper bound on the \emph{dependency degree} of $H$, i.e.,  $\max_{v\in V} \deg_H(v)\leq d$. 
The relationship between $p$ and $d$ is referred to as the \emph{LLL criterion}. Lovász showed that for $ep(d+1)<1$ there exists an assignment of the variables avoiding all the bad events \cite{EL74}. In the constructive version of the lemma we aim to compute such a bad-event avoiding assignment of the variables (but in our case under slightly stronger criteria).

Due to the following observation, we may assume that all random variables are fair coins. 
\begin{lemma}
    \label{lem:approx-vars}
    Let $B = (V_{\cE},V_{\var},E)$ be an LLL with probability $p$ and dependency degree $d$, where each variable in $V_{\var}$ has an enumerable (possible infinite) support.

    Each variable in $V_{\var}$ can be replaced by a finite set of binary random variables resulting in an LLL $B' = (V_{\cE},V'_{\var},E)$ with the same dependency degree but increased probability $2p$.

    When variables in $V_{\var}$ are such that for all of them, each element in their support has a minimum probability of $\delta$, they can each be replaced by at most $\ceil{4\Delta_{\cE}/p\delta}$ fair coins, where $\Delta_{\cE}$ is the maximum degree of each event, i.e., the maximum number of variables that each event depends on..
\end{lemma}

\begin{proof}
    Let $\Delta_{\cE}$ be the maximum degree of each event, i.e., the maximum number of variables that each event depends on. For each variable $X_u$, consider its support $\supp(X_u)$. Order the values that can be taken by $X_u$ in order of decreasing probability: $\Pr[X_u = x_i] \geq \Pr[X_u = x_j]$ for any $x_i,x_j \in \supp(X_u)$ s.t.\ $i \leq j$. 
    For any $\delta$ and variable $X_u$, let $\supp_\delta(X_u) = \set{x \mid \Pr[X_u = x] \leq \delta}$.
    
    Let $\delta$ and $\eps_{\delta,u}$ be such that: $\Pr[X_u \in \supp_\delta(X_u)] \leq \eps_{\delta,u}$
    
    Replace sampling according to $X_u$ by another random variable $\Xtilde_u$ sampled as follows:
    \begin{itemize}
        \item sample $k$ independent fair coins $Y_1,\ldots,Y_k$, compute $Z = \sum_{i=1}^k 2^{-i}Y_i \in 2^{-k}\set{0,\ldots,2^k-1}$.
        \item $\Xtilde_u = x_i$ where $i$ is the minimum index such that $\sum_{j=1}^i \Pr[X_u = x_i] \geq Z$.
    \end{itemize}
    By taking $k = \ceil{ \delta \eps_{\delta,u}}$, for every $x \in \supp(X_u)$, $\abs{\Pr[X_u = x] - \Pr[\Xtilde_u = x]} \leq 2^{-k} \leq \delta \eps_{\delta,u}$.
    Since values from $\supp_\delta(X_u)$ have probability at most $\eps_{\delta,u}$ when taken all together, and there are $\leq 1/\delta$ values outside of it,
    $X_u$ and $\Xtilde_u$ are $\eps_{\delta,u} + (1/\delta) (\delta \eps_{\delta,u})=2\eps_{\delta,u}$-close in total variation distance.
    By taking $\delta$ small enough that $\max_{u\in V_{\var}} \Pr[X_u \in \supp_\delta(X_u)] \leq p/(2\Delta_\cE)$, we get that for every variable $u$, $\Xtilde_u$ and $X_u$ are $p/\Delta_\cE$-close in total variation distance. 
    As a result, replacing all random variables $X_u$ by their approximation $\Xtilde_u$ results in a probability of each event $\cE_v$ happening increases by at most $\Delta_\cE \times p/\Delta_{\cE} = p$.

    Replacing each variable $X_u$ by the fair coins sampling $\Xtilde_u$, we get an LLL instance $B = (V_{\cE},V'_{\var},E)$ with the same dependency degree $d$, increased probability $2p$, and significantly more random variables.

    In particular, when there exists some value $\delta>0$ s.t.\ $\forall u \in V_{\var}, \min_{x \in \supp(X_u)}\Pr[X_u = x] \geq \delta$, then each variable $X_u$ can be replaced by $\ceil{4\Delta_{\cE}/p\delta}$ fair coins.
\end{proof}

\subsection{Fischer--Ghaffari LLL algorithm}
The distributed Fischer--Ghaffari algorithm for the \Lovasz Local Lemma is inspired by an earlier sequential algorithm for solving LLLs by Molloy and Reed \cite{MR_stoc98}, with the main difference being the amount of parallelism in the schedule for sampling variables. In this section, we present the original algorithm by Fischer and Ghaffari in which many variables are sampled simultaneously. Our analysis of the occurring shattering phenomenon presented in \Cref{sec:FGanalysis} works for a more general class of scheduling strategies including the algorithm by Molloy and Reed. Hence, it also establishes the missing\footnote{More precisely, Molloy and Reed state that the shattering argument in their algorithm is identical to an argument of Noga Alon for parallel algorithms for an LLL-type hypergraph coloring problem \cite{Alon91}. The crucial difference is that in Alon’s setting all variables are sampled simultaneously, which precludes the formation of long dependency chains affecting participation in the post-shattering phase. Consequently, Alon’s argument does not carry over to the sequential sampling used by Molloy and Reed. Such long-range dependencies can, for instance, be observed using the example from \Cref{example:FG}.} shattering argument in their work. 

\paragraph{Assignments and Marginal event probabilities.} Let us call functions $\alpha: \var(H) \to \set{0,1,\bot}$ \emph{assignments}, $\bot$ representing the unassigned state. An assignment is \emph{partial} if $\bot\in \alpha(H)$ and \emph{total} otherwise. Let $\emptyAs$ be the empty assignment, i.e., the assignment where all variables are mapped to $\bot$.  Let $\As$ be the set of all assignments (of variables).
The marginal probability of an event $\cE$ under some (partial) assignment $\alpha$ is the probability of the event to occur if all variables in $\alpha^{-1}(\bot)$ are sampled according to their distribution.

\paragraph{Algorithm (with threshold parameter $q$ to become dangerous).} Initially all variables are unset and also non-frozen. In the pre-shattering phase we will gradually set more and more variables, but also freeze some variables. The post-shattering phase then consists of all frozen variables together with their incident events. We next describe the pre-shattering phase: We begin with computing a distance-$2$ coloring of the dependency graph $H$ with $d^2$ colors. The coloring is used as a schedule for processing events. To this end, we iterate through the color classes, processing independently and in parallel all events of each given color. If an event $\cE$ of color $i$ is not frozen, it sequentially goes through all not yet frozen and not yet decided variables in $\var(\cE)$. When processing a variable it is sampled according to its distribution. Then, the node locally computes the marginal probability of all events using that variable. If the marginal probability of any event $\cE'$ increases to $>q$, the event is called \emph{dangerous} and all undecided variables in $var(\cE')$ are immediately frozen. 
This process continues until all variables of all events with color $i$ are set or frozen. Then we continue with the events in next color class.  See \Cref{alg:LLLShattering} for pseudocode of the algorithm.  

The post-shattering instance consists of all frozen undecided variables and all events with at least one such variable. \Cref{lem:LLLpostShattering} shows that this instance is an LLL with dependency degree $d$ and error probability bound $2q$.

\begin{algorithm}[ht]
\textbf{Initialization: }$F=\emptyset$; For all events $\cE: \var^{\mathsf{undecided}}(\cE)=\var(\cE)$\\
\textbf{Pre-shattering:}
Compute a distance-2-coloring $\psi$ of the event graph $H$ with $d^2+1$ colors \\
For $i=1$ to $d^2+1$, at each event $\cE$ s.t\ $\psi(v_{\cE})=i$:\\
 Sample  all variables in $\var^{\mathsf{undecided}}(\cE)\setminus F$ \underline{one by one}. 
    If any incident event $\cE'$ becomes dangerous (its marginal probability is at least $q$) immediately set $F=F\cup \var^{\mathsf{undecided}}(\cE')$

\textbf{Post-shattering instance: } Solve the marginal LLL $(N(F), F)$.\\
\textbf{Return:} Variable assignments of Pre-shattering and Post-shattering phases.
\caption{LLL algorithm (model independent)}
\label[algorithm]{alg:LLLShattering}
\end{algorithm}

In this process variable assignments, once sampled are permanent and never reverted. Variables once frozen stay frozen until the post-shattering phase. Thus once an event become dangerous none of its variables changes its status thereafter (in the pre-shattering phase).  

When all random variables associated with an independent set of events in the dependency graph $H$ are sampled simultaneously, the corresponding events occur independently. In contrast, under the gradual sampling and freezing process of the FG-LLL algorithm, the events of becoming frozen are generally not independent for distant nodes, except at distances of $\omega(d^2)$. An explicit example illustrating this phenomenon is given in \Cref{app:examples}.

\paragraph{Solvability of the post-shattering instance. }
Note that, in order for the post-shattering instance to remain solvable without retracting (i.e., unsetting) variable assignments, it is essential that the variables of each event are sampled sequentially. If multiple variables of an event are sampled simultaneously, the marginal probability of the event may change drastically, potentially to the point where the remaining instance cannot be solved efficiently in the \LOCAL model, or even at all. 

In contrast, sequential sampling ensures that marginal probabilities evolve in a controlled manner. In particular, when all variables are fair coins (see \Cref{lem:approx-vars}), the marginal probability cannot overshoot the threshold for becoming dangerous by more than a factor of $2$, resulting in a manageable residual instance for the post-shattering phase.
\begin{lemma}
\label{lem:LLLpostShattering}
The post-shattering instance is an LLL with dependency degree $d$ and bad event probability bound $2q$.
\end{lemma}
\begin{proof}
The post-shattering instance consists of events that became dangerous throughout the pre-shattering phase and may also consist of some events that did not become dangerous. An event that did not become dangerous has marginal probability (randomness over its frozen variables as all other variables are already set after the pre-shattering phase) $\leq q$. Any dangerous event $\cE$ crossed the threshold for becoming dangerous when one of its variables $x\in var(\cE)$ got sampled. Before $x$ got sampled the marginal probability of $\cE$ was less than $q$. Thus, after setting the binary variable $x$ to an adversarial value its marginal probability can have increased at most by a factor 2, since we assumed variables are fair coins.
After becoming dangerous the event $\cE$ and all of its variables are frozen and no further variables are set in the pre-shattering phase. Hence, its marginal probability in the post-shattering phase is at most $2q$. This reasoning requires  the schedule to be a distance-$2$ coloring of the dependency graph implying that the marginal probability of each event can only be influenced by a single event processed simultaneously in that color class. 
\end{proof}

\subsection{Shattering Analysis of the FG-LLL algorithm}
\label{sec:FGanalysis}
We next present the core technical contribution in analyzing the LLL pre-shattering phase by showing that the probability of any $3$-independent set to proceed to the post-shattering phase is exponentially small in the size of the set. The argument presented here is largely independent of both the model of computation considered and to a large extent the order in which variables are sampled.

To this end we continue with the following definitions for a given LLL dependency graph $H$.

To formalize the sampling procedures considered in this work, we model them as deterministic processes that adaptively select which variables to sample based on the current partial assignment. Note that while the choice of which variables to sample in the next step is deterministic for any given current partial assignment, the outcome of the process is still random, but the randomness comes purely from the random sampling of the variables. 
\begin{definition}[Deterministic sampling process]
\label{def:detSampling}
Let $\As$ be the set of assignments. A \emph{deterministic sampling process} $\cP: \As\to \As$ is defined by a \emph{choice function}
\[
f_{\cP} : \As \to 2^{\var(H)}
\]
that maps each assignment $\alpha$ to a subset of variables $f(\alpha) \subseteq \alpha^{-1}(\bot)$.

Starting from an initial assignment $\alpha_0$, the process proceeds iteratively as follows. At time step $t$, given the current assignment $\alpha_t$, let $U_t = f(\alpha_t)$.
\begin{itemize}
    \item If $U_t = \emptyset$, the process terminates and outputs $\cP(\alpha_0)=\alpha_t$.
    \item Otherwise, all variables in $U_t$ are sampled independently according to their distributions, yielding an updated assignment $\alpha_{t+1}$.
\end{itemize}
\end{definition}

\begin{observation}
\label{obs:stepReconstruction}
    For a fixed process $\cP$, some partial assignment $\alpha$ and the realization of the (partial) assignment $\cP(\alpha)$ one can uniquely reconstruct in which order the process $\cP$ assigned the variables.
    This is because the two assignments allow to perfectly reconstruct the history of the sampling of $\cP(\alpha)$: $\alpha$ defines which set of variables $f_\cP(\alpha)$ are sampled first, and the value of $\cP(\alpha)$ defines the outcomes of their sampling; this argument can be iterated over the whole process. 
\end{observation} An important idea in the proof is that the order in which variables are sampled does not affect the probability of getting a particular outcome.
This intuition is formalized in the next lemma. To state it formally, we require one definition. For a partial assignment $\alpha\in \As$ let 
$\mathcal{A}_{\mathcal{P}}(\alpha)=\{\beta \mid \Pr[\mathcal{P}(\alpha)=\beta]>0\}\subseteq \mathcal{A}$ be all configurations that process $\mathcal{P}$ can output from an initial assignment $\alpha$. 

\begin{lemma}
\label{lem:processProbability}
Consider any deterministic sampling process $\mathcal{P}$. 
Let $\alpha$ be an assignment.
Then
\begin{align}
    \label{eq:no-effect-of-order}
    \forall \beta:~
    \Pr[\mathcal{P}(\alpha)=\beta]
    =
    \begin{cases*}
    \prod_{x\in \alpha^{-1}(\bot) \setminus \beta^{-1}(\bot)}\Pr[x=\beta(x)]& if $\beta \in \As_{\mathcal{P}}(\alpha)$
    \\
    0 &~otherwise.
    \end{cases*}
\end{align} 
\end{lemma}
\begin{proof}
Consider a deterministic sampling process $\cP$ with function $f_{\cP} : \As \to 2^{\var(H)}$.

We define the \emph{(sampling) step potential} $h(\alpha)$ of an assignment $\alpha\in \As$ as follows: $h(\alpha) = 0$ iff $f_{\cP}(\alpha) = \emptyset$, otherwise, $h(\alpha) = 1 + \max_\beta p(\beta)$ where $\beta$ is taken over the direct successors of $\alpha$, that is, all assignments that can be reached from $\alpha$ with a single sampling step of the variables in $f_{\cP}(\alpha)$. Note that the step potential is well defined: One can define a directed acyclic graph with a directed edge from any assignment $\alpha$ such that $f_\cP(\alpha) \neq \emptyset$ to assignments $\alpha'$ that can be obtained from $\alpha$ with a single (parallel) sampling step of the variables in $f_{\cP}(\alpha)$. For every directed edge $\overrightarrow{\alpha\alpha'}$, $\alpha'$ has strictly more assigned variables than $\alpha$, guaranteeing acyclicity. The step potential of an assignment with no outgoing edge is $0$, and the step potential of any other assignment is the length of the longest path to an assignment with no outgoing edge. 

\Cref{eq:no-effect-of-order} holds trivially for assignments $\alpha$ with step potential $h(\alpha)=0$. We show it to hold for all assignments by induction over the sampling potential. For some $i$, assume that \Cref{eq:no-effect-of-order} holds for all assignments $\alpha$ with step potential $h(\alpha)\leq i$
For the induction step, consider some assignment $\alpha$ with step potential $i+1$. For any $\beta \not\in \As_{\cP}(\alpha)$ we have $\Pr\event{\cP(\alpha)=\beta}=0$, by the definition of $\As_{\cP}(\alpha)$; thus \Cref{eq:no-effect-of-order} holds for these $\beta$. For any assignment $\beta\in\As_{\cP}(\alpha)$,
let $\beta'$ be the restriction of $\beta$ on the variables in 
$(\var(H)\setminus \alpha^{-1}(\bot))  \cup f_{\cP}(\alpha)$~.
As mentioned in \cref{obs:stepReconstruction}, for $\cP(\alpha)$ to produce $\beta$, it first needs to sample the variables in $f_\cP(\alpha)$ and obtain $\beta'$, before continuing the progressive sampling to obtain $\beta$. 
We have:

\begin{align}
\Pr[\cP(\alpha)=\beta] & =
\parens*{
\prod_{x\in f_{\cP}(\alpha)}\Pr[x=\beta(x)]
}\cdot \Pr\event{\cP(\beta')=\beta}~. 
\end{align}
With the induction hypothesis applied to $\beta'$ ($\beta'$ has strictly smaller step potential), we get the claim also for $\beta\in \As_{\cP}(\alpha)$. 
\end{proof}

The only assumptions on the progressive sampling of variables that we require for a successful pre-shattering phase are captured by the following definition:
\begin{definition} 
\label{def:samplingProcess}
A \emph{granular sampling process $\mathcal{P}:\mathcal{A}\rightarrow \mathcal{A}$ with freezing-threshold $q$}  is a deterministic sampling process satisfying the  following properties: 
\begin{enumerate}
    \item An event has at most one of its variables sampled at a time:
   
    \emph{``\,for any event $\cE$ and any partial assignment $\alpha$ we have $|f_{\cP}(\alpha)\cap \var(\cE)|\leq 1$\,'',}
    \item Whenever an event has marginal probability $\geq q$, all its unsampled random variables are \emph{frozen}, i.e., will not get sampled:

 \emph{``\,for each partial assignment $\alpha$ and each event $\cE$ with $\Pr[\cE \mid\alpha]\geq q$ we have $f_{\cP}(\alpha)\cap \var(\cE)=\emptyset$\,''}, and
    \item All non-frozen random variables are eventually sampled:

    \emph{``\,for all partial assignments $\alpha$ and all variables $v$ that evaluate to $\bot$ in $\cP(\alpha)$, there exists some event $\cE$ with $v\in \var(\cE)$ satisfying $\Pr\event{\cE \mid \cP(\alpha)}\geq q$.\,''}
\end{enumerate}
\end{definition}
Note that in a granular sampling process no variable is ever retracted. 
Note that freezing variables in such a process does not change the marginal probability of events. \Cref{alg:LLLShattering} is an example of a granular sampling process with freezing-threshold $q$, see \Cref{obs:FGGranular} for details.

Before proving the lemma, let us highlight the main idea. We do not try to argue that the events of different nodes entering the post-shattering instance are independent; in fact, this is precisely what can fail. Instead, we condition on the partial assignment produced by the adaptive sampling process. If all events in a $3$-independent set $S$ enter the post-shattering instance, then in each disjoint neighborhood around a node of S, some event has reached marginal probability at least $q$. Thus, when sampling the remaining variables in parallel, all these neighborhoods contain a bad event with probability at least $q^{|S|}$.  On the other hand, under a fresh full parallel sampling from scratch, the probability that all these disjoint neighborhoods contain a bad event is at most $((d+1)p)^{|S|}$: $p^{|S|}$ is the probability that $|S|$ independent bad events all hold, and $(d+1)^{|S|}$ is the number of ways to choose an event in each of $|S|$ $1$-hop neighborhoods. Comparing these two bounds via the law of total probability yields the desired $\parens*{\frac{(d+1)p}{q}}^{|S|}$ bound.
 
\begin{lemma}
\label{lem:smallSetProbability}
Let $0\leq q\leq 1$ and let $H=(V_{\cE},E)$ be the dependency graph of an LLL with error probability bound $p$ and dependency degree $d$.  Let $S\subseteq V_{\cE}$ be a $3$-independent set of nodes in $H$. Let $B$ be the random variable describing the set of events with an incident frozen variable after running a granular sampling process with freezing-threshold $q$ on $H$. Then $\Pr[S\subseteq B]\leq \big(\frac{(d+1)p}{q}\big)^{|S|}$. 
\end{lemma}
\begin{proof}
Let $\AsFull$ be the set of \emph{full assignments}, i.e., assignments in which no variable is mapped to $\bot$, and let $\ProFull$ be the process which samples all unsampled variables (as fair coins), thus mapping all $\alpha \in \As$ to some $\ProFull(\alpha) \in \AsFull$. For some set $X$ of events let $\AsFail{X} \subseteq \AsFull$ be the set of full assignments in which all the (bad) events of $X$ occur in the full assignment. Let $\AsFailp{X} \subseteq \AsFull$ be the set of full assignments in which each node $v \in X$ has an event fail in its inclusive neighborhood $\set{v}\cup N(v)$.

Consider a $3$-independent set $S$ in $H$ and let $k=|S|$. Let us analyze what different processes $\Pro:\As \to \As$ do to this set and its neighborhood.
In a total assignment of the random variables, whether an event node $v\in S$ or one of its neighbors in the dependency graph fails depends on the variables $\var(\set{v} \cup N(v)) = \var(v) \cup \bigcup_{u \in N(v)} \var(u)$. For nodes in $S$, these sets are mutually disjoint by 3-independence, and the corresponding events are therefore mutually independent.
Each of those neighborhood-events have probability at most $(\deg(v)+1)p \leq (d + 1)p$ by union bound over the inclusive neighborhood, implying the following bound on the probability that they occur for all nodes in $S$:
\begin{equation}
\label{eq:indep-failure-full-sampling}
\Pr\event*{\ProFull(\emptyAs) \in \AsFailp{S}} 
= \prod_{v \in S} \Pr\event*{\ProFull(\emptyAs) \in \AsFailp{\set{v}}}
\leq ((d+1) p)^k
\ .
\end{equation}

To prove the lemma consider a granular sampling process $\ProStop{q}$ with freezing threshold $q$. Which variables are sampled and the order of their sampling are not specified beyond the properties of \Cref{def:samplingProcess}, but the process is fixed for the rest of the proof. 
Let $\AsStop{q}\subseteq \mathcal{A}$ be the set of assignments that can be generated by the process $\ProStop{q}$ starting from the empty assignment $\emptyAs$.
Let $\AsStopp{q,S}\subseteq \AsStop{q}$ be the subset of those assignments in which all nodes in $S$ had an event $u$ in their inclusive  neighborhood reach marginal probability $q$. Formally:
\begin{align*}
\forall \alpha \in \AsStopp{q,S}, \forall v \in S, \exists w_{v,\alpha} \in \set{v} \cup N(v): &&  \Pr\event*{\ProFull(\alpha) \in \AsFail{\{w_{v,\alpha\}}}} \geq q
\ .\end{align*}

Let $\alpha \in \AsStopp{q,S}$, where $S$ is again a set of $k$ nodes, each at distance at least $4$ from other nodes in $S$. Let us analyze the probability that sampling its remaining variables leads to an event failing in the inclusive neighborhood of each element of $v$. For each $v \in S$, let $w_{v,\alpha} \in \set{v} \cup N(v)$ be an event whose marginal probability reached $q$ which exists as $\alpha \in \AsStopp{q,S}$. Again, as when establishing \cref{eq:indep-failure-full-sampling}, the events $w_v$ have disjoint sets of random variables left to sample. Thus we have for all $\alpha \in \AsStopp{q,S}$: 
\begin{equation}
\label{eq:indep-failure-finish-sampling}
\Pr\event*{\ProFull(\alpha) \in \AsFailp{S}} 
= \prod_{v \in S} \Pr\event*{\ProFull(\alpha) \in \AsFailp{\set{v}}}
\geq \prod_{v \in S} \Pr\event*{\ProFull(\alpha) \in \AsFail{\set{w_{v,\alpha}}}}
\geq q^k
\ .
\end{equation}

Note that the outcome  $\ProStop{}(\emptyAs)$ of the process $\ProStop{}$ is a random variable. 
 The first two equalities in the following calculation apply \Cref{lem:processProbability} several times first to determine $\Pr[\ProFull(\emptyAs)=\alpha]$, and then for the second euality with $\Pr[\ProStop{q}(\emptyAs) = \alpha]$ and  $\Pr[\ProFull(\alpha)=\alpha']$ for all $\alpha\in \mathcal{A}$ and all $\alpha'\in \AsFailp{S}$.
\newlength{\padLength}\newlength{\padA}\newlength{\padB}
\settowidth{\padA}{$\scriptstyle\alpha \in \As$}
\settowidth{\padB}{$\scriptstyle\alpha \in \AsStopp{q,S}$}
\setlength{\padLength}{(\padA*2+\padB)/6}\newcommand{\fixWidth}[1]{\rule{\padLength}{0pt}\mathclap{#1}\rule{\padLength}{0pt}}
\allowdisplaybreaks
\begin{align}
\Pr\event*{\ProFull(\emptyAs) \in \AsFailp{S}} \notag
    & = \sum_{\fixWidth{\alpha \in \As}}  \prod_{x\in var(H)\setminus \alpha^{-1}(\bot)} \Pr[x=\alpha(x)] \cdot 1_{{\alpha\in \AsFailp{S}}}\\
    & = \sum_{\fixWidth{\alpha \in \As}} \Pr\event*{\ProStop{q}(\emptyAs) = \alpha} \cdot \Pr\event*{\ProFull(\alpha) \in \AsFailp{S}} \notag\\ 
    &= \sum_{\fixWidth{\alpha \in \AsStopp{q,S}}} \Pr\event*{\ProStop{q}(\emptyAs) = \alpha} \cdot \Pr\event*{\ProFull(\alpha) \in \AsFailp{S}} \notag\\
    &\quad+ \sum_{\fixWidth{\alpha \not\in \AsStopp{q,S}}} \Pr\event*{\ProStop{q}(\emptyAs) = \alpha} \cdot \Pr\event*{\ProFull(\alpha) \in \AsFailp{S}} \notag\\
    & \geq \sum_{\fixWidth{\alpha \in \AsStopp{q,S}}} \Pr\event*{\ProStop{q}(\emptyAs) = \alpha} \cdot \Pr\event*{\ProFull(\alpha) \in \AsFailp{S}} + 0 \notag\\
&\stackrel{(\ref{eq:indep-failure-finish-sampling})}{\geq} q^k \cdot \parens[\Big]{\sum_{\fixWidth{~~\alpha \in \AsStopp{q,S}}} \Pr\event*{\ProStop{q}(\emptyAs) = \alpha}}  \notag\\
    &=  q^k \cdot \Pr\event*{\ProStop{q}(\emptyAs) \in \AsStopp{q,S}}  \ ,\end{align}
Dividing by $q^k$, we obtain:
\[ \Pr\event*{\ProStop{q}(\emptyAs) \in \AsStop{q,S,1}}
\leq \frac{\Pr\event*{\ProFull(\emptyAs) \in \AsFail{S,1}}}{q^k}
\stackrel{(\ref{eq:indep-failure-full-sampling})}{\leq} \frac{((d+1)\cdot p)^k}{q^k}
\ .\]
\end{proof}

\subsection{Distributed LLL: Proof of Theorem~\ref{thm:mainLLL}}
\paragraph{Distributed Lovász Local Lemma and Implementation of \Cref{alg:LLLShattering}.}
Formally, in the distributed version of the constructive Lovász Local Lemma each event and each variable is assigned to some node of the communication network. See \cite{MPU23,Davies_soda23} for formal definitions. 
In the \LOCAL model it is generally assumed that the dependency graph of the LLL also serves as the communication network. 
The objective remains to compute an assignment of the variables avoiding all bad events. We next show that \Cref{alg:LLLShattering} can be efficiently implemented in the \LOCAL model and prove Theorem~\ref{thm:mainLLL}.

\begin{lemma}
\label{lem:preshatteringRuntime}
The pre-shattering phase can be run in $O(d^2)+O(\log^*n)$ rounds in the \LOCAL model. 
\end{lemma}
\begin{proof}
The coloring $\psi$ can be computed in $O(d^2)+O(\log^*n)$ rounds, e.g.\ using \cite{BEG17,M21,FHK,barenboim15}. 
 Dealing with a single color class can be done in $O(1)$ rounds as the sequential sampling of the still unset variables in $var(\cE)$ is done locally and simulated by the node holding $\cE$. To this end, an event node $\cE$ gathers full information of its $O(1)$-neighborhood,  including its incident events' definitions, freezing the state of events and variables and partial assignments of variables. 
In each round, each event has variables set by a single event, and thus not conflicting with the marginal probabilities w.r.t.\ other events.
\end{proof}

\begin{observation}
\label{obs:FGGranular}
\Cref{alg:LLLShattering} is a granular sampling process.
\end{observation}
\begin{proof}
Fix the distance-$2$ coloring $\psi$ of the dependency graph $H$ and, for each event, fix an arbitrary
order of its incident variables. We view Algorithm~1 as a deterministic sampling process whose choice
function $f_P$ is defined as follows. Given a partial assignment $\alpha$, together with the set of already
frozen variables determined by $\alpha$, let $i$ be the first color class that has not yet been completely
processed. For every event $E$ of color $i$ that is not frozen, include in $f_P(\alpha)$ the next variable of
$\operatorname{var}_{\mathrm{undecided}}(E)$ that is not frozen and has not yet been sampled, if such a variable
exists. If no such variable exists for any event of the current color class, the process advances to the next
color class; if all color classes have been processed, then $f_P(\alpha)=\emptyset$.
It  is straightforward to verify the properties of a granular sampling process with this choice function.

The described granular sampling process splits a round into (possibly) many further sampling steps. However, the process itself has no intra round synchronization between nodes in the same color class. But since $\psi$ is a distance-$2$
coloring of $H$, two events of the same color have no common neighboring event and, in particular,
their sampled variables cannot affect the marginal probability of a common event. Thus, given the distance-$2$ coloring and the internal order of the variables of each events,  the parallel 
steps of \Cref{alg:LLLShattering} are linearizable into the above deterministic choice function without changing the
resulting probabilistic partial assignment or the set of frozen variables.
\end{proof}

\ifx\thetmpCounter\undefined
\newcounter{tmpCounter}
\fi
\setcounter{tmpCounter}{\thetheorem}
\setcounter{theorem}{\thesavedLLL}

\addtocounter{savedLLL}{1}
\edef\temp{theorem.\thesavedLLL.complex}
\addtocounter{savedLLL}{-1}
\hypertarget{\temp}{
\begin{restatable}[\cite{FG17} redux (technical)]{theorem}{thmLLLcomplex}
\customlabel{thm:LLL}{complex}
For any integer $d \geq 1$, $0<q<\frac{1}{e(d+1)}$, and any constant $\eps>0$, there is a randomized distributed algorithm for the constructive \lovasz Local lemma under criterion $p \cdot\frac{2(d+1)d^{6+\eps}}{q}<1$ that w.h.p.\ runs in $O(d^2)+\Det_{\mathsf{LLL}, q,d}(\poly\log n)$ rounds. 
\end{restatable}
}
\setcounter{theorem}{\thetmpCounter}

Applying this more complex version of Theorem~\ref{thm:mainLLL} with $q=\frac{1}{3(d+1)}$, $\eps = 0.1$ and simple calculations yields the simplified form for criterion $8pd^{-9} < 1$ presented earlier in the document. Indeed, $8pd^{-9} < 1$, $q=\frac{1}{3(d+1)}$, and $\eps = 0.1$ imply
$p \cdot\frac{2(d+1)d^{6+\eps}}{q} < \frac{1}{8d^9} \cdot 6(d+1)^2d^{6.1} = \frac{3(d+1)^2}{4 d^{2.9}}$ which is $<1$ for all $d \geq 2$, enabling the use of this more technical form. The trivial case $d=1$ can be treated separately.

\begin{proof}[Proof of Theorem~\ref{thm:mainLLL}]
For $d=1$, the result holds trivially as the existence of a solution is still guaranteed and can be computed in a single round. 
If $d= \Omega(\log n)$, use the randomized algorithms from \cite{MoserTardos10,CPS17} to solve the LLL instance in $O(\log n)=O(d^2)$ rounds with high probability. 

If $d=O(\log n)$, use the algorithm described in this section with  threshold parameter $q/2$ where $q$ is the value of Theorem~\ref{thm:mainLLL}.
By \Cref{obs:FGGranular,lem:preshatteringRuntime} the runtime of the pre-shattering phase is $O(d^2+\log^* n)$. 
By \Cref{lem:smallSetProbability} the probability of any $3$-independent set $S$  to participate in the post-shattering instance is at most  
\begin{align*}
(2(d+1)p/q)^{|S|}\leq d^{-(6+\eps)|S|},
\end{align*}
using the assumption $p\leq q(d+1)^{-1}d^{-6-\eps}/2$ from the theorem statement. Via \Cref{cor:shattering-small-cc} (with $c_2=3$, $c_3=\eps$, $t=\eps^{-1}\log_{\Delta} n$) we obtain that connected components in the post-shattering instance are of size $N=t\cdot \Delta^{c_2}=\eps^{-1}\poly\Delta\log_{\Delta} n=\poly\log n$ with probability at least $1-\Delta^{-c_3t}\geq 1-1/n$. By \Cref{lem:LLLpostShattering} the post-shattering instance is an LLL with dependency degree $\leq d$ and bad event probability bound at most $2\cdot (q/2)=q$ that can be solved with a deterministic LLL algorithm, establishing the theorem. \end{proof}

LLLs can be solved efficiently with deterministic algorithms in the \LOCAL model; in the corollary below we use the following lemma in the post-shattering phase to bound $\Det_{\mathsf{LLL}, q,d}(N)$. \Cref{lem:deterministicLLLLOCAL} is proven by using the powerful general derandomization framework of \cite{RG20,GHK18,GKM17} for the algorithm of Moser-Tardos \cite{MoserTardos10} when the criterion is $epd(1+\eps)<1$ and with the randomized LLL algorithm of Chung, Pettie, and Su \cite{CPS17} for criterion $pd^2<1$.
\begin{lemma}[Deterministic LLL in \LOCAL, \cite{RG20,GG24-FOCS,CPS17}]
\label{lem:deterministicLLLLOCAL}
There are  deterministic \LOCAL algorithms for the constructive Lovász local lemma with $n$ events if node IDs are from a space of size $s$ with the following runtimes:
\begin{enumerate}
    \item $O(\log^* s)+\Otilde(\log^4 n)$ if $epd(1+\eps)<1$ holds,
    \item $O(\log^* s)+\Otilde(\log^3 n)$ if  $epd^2<1$ holds. 
\end{enumerate}
\end{lemma}

We obtain the following corollary by plugging suitable values for $q$ and the respective runtimes of \Cref{lem:LLLpostShattering} into Theorem~\ref{thm:mainLLL}. 
\begin{corollary}
\label{cor:LLLcorollary}
For any constants $\eps_1, \eps_2>0$ there are  randomized \LOCAL algorithms for the constructive Lovász local lemma with $n$ events and dependency degree $d$ that w.h.p.\ have the following runtimes
\begin{enumerate}
    \item $O(d^2)+\Otilde(\log^4\log n)$ rounds if $2ep(d+1)d^{7+\eps_1}(1+\eps_2)<1$ holds,
    \item $O(d^2)+\Otilde(\log^3\log n)$ rounds if  $2ep(d+1)d^{8+\eps_1}<1$ holds. 
\end{enumerate}
\end{corollary}

This completes the shattering analysis for the LLL algorithm, demonstrating that our framework successfully handles the dependencies that invalidate standard approaches.

\section{Shattering with Dependent Opportunities}
\label{sec:general-opportunity-proof}
\label{sec:opportunity}

We consider a general strategy that we term \emph{opportunistic process} in which nodes accumulate incremental progress across multiple iterations. 
A node is declared \emph{solved} once it has achieved at least $t$ units of progress. 
Each node is guaranteed to experience at least $x$ \emph{opportunity rounds}, during which its probability of making (one unit of) progress is at least $q$. 
This is independent between different nodes (typically due to separation), conditioned on past rounds.

\begin{definition}\label{def:opportunistic-process}
    Let $x$, $t$, and $c$ be non-negative integers, and $q\in [0,1]$.
    A $(x,q,t,c)$-\emph{progressive opportunity process}
is an algorithm with the following properties:
    \begin{enumerate}
        \item Each node $v$ has a progress counter, initially $0$, which is non-decreasing over time. Once the counter of $v$ reaches $t$, the node $v$ is guaranteed to be solved by the algorithm.
        \label{def-prop:opportunity-progress-counter}
        \item The algorithm is divided into steps, each consisting of one or more rounds.
        Consider a step and condition on all randomness from previous steps. Under this conditioning, some nodes have an \emph{opportunity} during the step.
        A node with an opportunity has a probability of at least $q$ that its progress counter increases by at least $1$ during the step---the opportunity is said to be \emph{successful} if that happens.
\label{def-prop:opportunity-min-prob}
        \item In a given step, conditioned on the randomness of previous steps, opportunities at nodes from a $c$-independent set $S$ succeed independently.
        \label{def-prop:opportunity-independence}
        \item Each node has at least $x$ opportunities over the whole process.
        \label{def-prop:opportunity-guaranteed-number}
    \end{enumerate}
\end{definition}

The parameter $x$ is the \emph{opportunity count}, $q$ is the 
\emph{progress probability}, $t$ is the \emph{solving threshold}, and $c$ is the \emph{independence radius}.
Whether a node has an opportunity during a step may correspond to some particular combinatorial configuration or state of the algorithm. For example, in the MIS algorithm of~\cite{BEPS16_jacm}, a node $v$ has an opportunity in steps in which none of its neighbors has a much higher degree than its own (or, for technical reasons, if it is already solved).

At a high level, analyzing an algorithm as an opportunistic process is a form of stochastic domination:
The set of nodes left unsolved by the algorithm is a subset of the nodes that came short of their threshold in the opportunistic process. Thus, putting some structural bounds on the set of nodes whose counter does not reach the threshold
also bounds the set of unsolved nodes contained within.

\smallskip

While the outcomes of nodes that are sufficiently far apart are independent, the \emph{timing} of their opportunities may be arbitrarily correlated. 
This dependence prevents the direct application of standard concentration bounds. 
Nevertheless, we show that for suitable relationships among the parameters $t$, $q$, and $x$, this process still exhibits a shattering phenomenon; see \Cref{sec:opportunity} for the formal statement and proof.

We will show the following by induction. 
Let $p = 1-q$ be the probability of failing one's opportunity, so that when $k$ nodes have an opportunity in a given round, the probability that all of them fail their opportunity is at most $p^k$.

\thmOpportunities*

A strength of \Cref{thm:probability-guarantee} is that it is relatively simple to apply in order to meet the requirements of the general shattering lemmas \Cref{lem:ShatteringGeneral,cor:shattering-ball-graph,cor:shattering-ball-graph}. In \Cref{sec:mis-ghaffari,sec:mis-beps} we apply it to show that the maximal independent set algorithms of \cite{Ghaffari_soda16,BEPS16_jacm} indeed shatter.

One result that this theorem does not recover is the algorithm for maximal matching of \cite{BEPS16_jacm}. Intuitively, while in algorithms for MIS the claim is that each node experiences $\Theta(\log \Delta)$ rounds where it has a constant probability of getting solved (from joining the MIS, or one of its neighbors joining the MIS), in the algorithm for maximal matching, the claim is that each node has at least a constant probability of its degree reducing by a constant factor. That is, in the argument for maximal matching, it is the accumulation of $\log \Delta$ reductions in degree rather than a single event that guarantees that the node is not part of post-shattering.
This calls for a more general statement, which we prove in \cref{ssec:generalizedOpportunityProcess}.

\subsection{Proof of \texorpdfstring{\Cref{thm:probability-guarantee}}{Theorem~\ref{thm:probability-guarantee}}}

Let us introduce some notation to prove \Cref{thm:probability-guarantee}. Throughout let $T$ be the number of rounds of the process and $x$ the number of guaranteed opportunities. Note that $x$ measures the opportunities of all nodes combined and we do not require any guarantee for individual nodes, even though such stronger guarantees hold in most applications.  
\begin{itemize}
\item For each $i \in [T+1]$ let $\Psi_i$ be the set of all possible states of the system prior to running round number $i$ (that is, in each individual state the randomness from round $i$ is not yet known, but that from all rounds $j < i$ is). As a basic level, a state in round $i$ is just a description of all the random choices prior to round $i$.

Throughout this section, for a given state $\alpha \in \Psi_i$, we will also treat $\alpha$ as the indicator random variable for reaching this state; similarly, we treat sets of states as the indicator random variable that one of the states in the set happen.

\item For a state $\alpha \in \Psi_i$ in round $i$, let $\succ(\alpha) \subseteq \Psi_{i+1}$ be the set of states in the round $i+1$ that can be reached from $\alpha$. Conversely, let $\pred(\alpha) \in \Psi_{i-1}$ be the unique predecessor state of $\alpha$, basically all the random choices made to get to $\alpha$ prior to round $i-1$.
  
\item For each $i \in [T+1]$ and $g \geq 0$ let $O_i^{(g)} \subseteq \Psi_i$ be the states that had exactly $g$ opportunities prior to round $i$, and for every $k \geq 0$, let $O_i^{(g)\uparrow k} \subseteq O_i^{(g)}$ be the subset of those states s.t.\ nodes in $S$ have $k$ opportunities in round $i$. We also write $O_i^{(\geq g)} = \bigcup_{h \geq g} O_i^{(h)}$ (resp., $O_i^{(< g)} = \Psi_i \setminus O_i^{(\geq g)} = \bigcup_{h < g} O_i^{(h)}$) for the set of states which had at least (resp., less than) $g$ opportunities pre round $i$.

\item For each $i \in [T+1]$, let $U_i \subseteq \Psi_i$ be the subset of states where all nodes of $S$ are still undecided at the beginning of round $i$ (or after round $i-1$).
\end{itemize}

Some elementary properties:
\begin{itemize}
\item $\alpha \in O_i^{(g)\uparrow k} \implies \succ(\alpha) \subseteq O_{i+1}^{(g+k)}$, or equivalently, $\Pr[O_{i+1}^{(g+k)} \mid O_i^{(g)\uparrow k}] = 1$.
\item $\alpha \in U_{i+1} \implies \pred(\alpha) \in U_{i}$, or equivalently, $\Pr[\neg U_{i+1} \mid \neg U_{i}] = 1$.
\end{itemize}
The last point in the list simply states that if all nodes are undecided in round some round they are also undecided in the previous round. In this notation, the guarantee given by the MIS algorithms of \cite{BEPS16_jacm} are $T = c \log^2 \Delta$, $\Pr[O^{(< c |S| \log \Delta)}_{T+1}] = 0$, and the guarantees given by Ghaffari's algorithm \cite{Ghaffari_soda16} are $T = 13 c \log \Delta$, $\Pr[O^{(< c |S| \log \Delta)}_{T+1}] = 0$. See \Cref{sec:mis-ghaffari,sec:mis-beps} for more details on these algorithms and their analysis.

The definition of an opportunity (to be decided) implies that when $k$ undecided nodes have an opportunity in a given round, the probability that each of them is undecided at the end of the round is at most $p^k$. In technical terms the definition of opportunities implies the following:
\begin{align}
  \forall \alpha \in O^{(g)\uparrow k}_i,\ \Pr[U_{i+1} \mid \alpha] \leq p^k
\end{align}

 Note that we even have $\Pr[U_{i+1} \mid \alpha] = 0$ when $\alpha \not \in U_i$.

\begin{proof}[Proof of \Cref{thm:probability-guarantee}]
In the introduced notation the theorem assumption states $\Pr[O^{(< x)}_{T+1}] = 0$ and we need to prove 
\begin{align}
\label{eqn:allUndecided}
    \Pr[U_{T+1}] \leq p^x\ .
\end{align}

  In order to prove \Cref{eqn:allUndecided},  we show the following more general statement by backwards induction on $i=T+1,\ldots, 0$, intuitively showing that starting from a state $\alpha \in O_i^{(x-k)}$ that still has $k$ opportunities to experience, the probability that the set $S$ is still unsolved by the end of the process is $\leq p^k$.
  \begin{align}
  \label{eqn:inductionHypothesis}
  \forall i \in [T+1], k \in \set{0,\dots,x}, \alpha \in O_i^{(x-k)},\quad \Pr[U_{T+1} \mid \alpha]\leq p^k\ .
  \end{align}

\textbf{Induction base:} For the induction base note that the statement is trivial for $i=T+1$: if $k=0$, \Cref{eqn:inductionHypothesis} holds trivially as $p^k = 1$; if $k \geq 1$, the set $O_{i}^{(x-k)} \subseteq O_{i}^{(<x)}$ is empty because of the theorem assumption ($\Pr[O_{T+1}^{(<x)}] = 0$) on the guaranteed number of opportunities experienced throughout the process.

  \textbf{Induction step:} Fix some $i$ and suppose \Cref{eqn:inductionHypothesis} holds for every $j \in [i+1,T+1]$. We prove the induction hypothesis for each $k\in \{0,\ldots,x\}$ and for each $\alpha\in O_i^{(x-k)}$.  
  Let $d\geq 0$ be such that the nodes of $S$ have $d$ opportunities in round $i$ in state $\alpha$. 
  We upper bound  $\Pr[U_{T+1} \mid \alpha]$ by $p^k$ via the following sequence of (in)equalities; the reasons for each individual step follow at the end of the proof. 
    \begin{align}
    \Pr[U_{T+1} \mid \alpha]
    & = \sum_{\beta \in \succ(\alpha)} \Pr[U_{T+1} \cap \beta \mid \alpha]
    \label{eq:opportunities-proof-succ-part}
    \\
    & = \sum_{\beta \in \succ(\alpha)\cap U_{i+1}} \Pr[U_{T+1} \cap \beta \mid \alpha]
    \label{eq:opportunities-proof-stays-solved}
    \\
    & = \sum_{\beta \in \succ(\alpha)\cap U_{i+1}} \Pr[U_{T+1} \mid \alpha \cap \beta]\cdot \Pr[\beta \mid \alpha]
    \label{eq:opportunities-proof-condition}
    \\
    & = \sum_{\beta \in \succ(\alpha)\cap U_{i+1}} \Pr[U_{T+1} \mid \beta]\cdot \Pr[\beta \mid \alpha]
    \label{eq:opportunities-proof-remove-alpha}
    \\
    & \leq \sum_{\beta \in \succ(\alpha)\cap U_{i+1}} p^{\max(0,k-d)}\cdot \Pr[\beta \mid \alpha]
    \label{eq:opportunities-proof-induction}
    \\
    & = p^{\max(0,k-d)} \cdot \Pr[U_{i+1} \mid \alpha]
    \notag
    \\
    & \leq p^{\max(0,k-d)} \cdot p^d \leq p^k\ .
    \label{eq:opportunities-proof-low-failure}
    \end{align}

As $k\in \{0,\ldots,x\}$ and $\alpha \in O_i^{(x-k)}$  were arbitrary, this concludes the induction step. We used in order:
  \begin{enumerate}
  \item The law of total probability in \cref{eq:opportunities-proof-succ-part},
  \item That when $S$ is no longer fully unsolved, it stays that way: $\neg U_i \implies \neg U_j$ for all $j > i$, and so $\Pr[U_{T+1} \cap \neg U_{i+1}] = 0$ in \cref{eq:opportunities-proof-stays-solved}.
  \item The standard formula $\Pr[A \cap B \mid C] = \Pr[A \mid B\cap C] \cdot \Pr[B \mid C]$ in \cref{eq:opportunities-proof-condition}, as indeed
    \[
    \Pr[A \cap B \mid C]
    = \frac{\Pr[A \cap B \cap C]}{\Pr[C]}
    = \frac{\Pr[A \mid B \cap C] \cdot \Pr[B \cap C]}{\Pr[C]}
    = \Pr[A \mid B\cap C] \cdot \Pr[B \mid C]
    \ .\]
  \item $\alpha \cap \beta = \beta$ when interpreted as events (of reaching the respective states during the process), since $\beta \in \succ(\alpha)$, in \cref{eq:opportunities-proof-remove-alpha}.
  \item The induction hypothesis in \cref{eq:opportunities-proof-induction}, together with the fact that probabilities do not exceed $1=p^0$.
  \item The bounded probability of failure of opportunities in \cref{eq:opportunities-proof-low-failure}, i.e., $\Pr[U_{i+1} \mid \alpha] \leq p^d$ because $\alpha \in O_i^{(x-k)\uparrow d}$ has $d$ opportunities in this round (and their failure probabilities can be multiplied by the definition of an opportunity).
  \end{enumerate}
  
\Cref{eqn:inductionHypothesis} applied to the initial state $\eps \in \Psi_1 = O_1^{(0)}$ and $k=x$, we get that $\Pr[U_{T+1} \mid \eps] \leq p^x$, proving \Cref{thm:probability-guarantee}. 
\end{proof}

\subsection{Generalized Opportunity Decision Process}
\label{ssec:generalizedOpportunityProcess}

\begin{theorem}
\label{thm:opportunityProcessGeneralized}
Consider an opportunity decision process where each node $v\in S$ of a set $S$ is guaranteed to have $x_v$ opportunities and the success per opportunity is $q_v$, then the probability that each $v\in S$ makes $t_v$ or less progress is at most 
 \begin{align}
     \prod_{v\in S} \Pr[B(x_v,q_v) \leq t_v), 
 \end{align}
 where $B(x_v,q_v)$ is the binomial distribution of parameters $x_v$ and $q_v$, i.e., the random variable representing the number of successful outcomes (heads) when flipping $x_v$ coins with bias $q_v$ towards head.
\end{theorem}
\begin{proof}
To be included in the full version. 
\end{proof}

The theorem should also hold with a slightly different definition of a progressive opportunity process: one in which 
\cref{def-prop:opportunity-independence} of \cref{def:opportunistic-process} is relaxed to only require that $c$-independent opportunities behave "as-if independent" rather than independently. That is, for the set of nodes that have an opportunity in a round, the random variables indicating whether these opportunities are successful should stochastically dominate a binomial distribution.

\section{Multi-Phase Shattering}
\label{sec:coloring}
\label{sec:multiPhase}
In \Cref{sec:opportunity}, a node enters the post-shattering phase if it is not successful in any of the rounds of the process. Often, in such a setting one can boost the probability  by running the process for longer. In this section we focus on a different setting, where the algorithm is split into different phases and in each phase nodes may \emph{fail} to meet the progress measure and end up in the post-shattering phase. 
Generally the goal of each phase is to achieve some preconditions for the next phases.
If one had a separate post-shattering instance for each of the phases, then each phase may fit into the prior framework. The downside is that then these may need to be solved sequentially. To obtain the full result, e.g., of \cite{CLP20}, one wants to combine the failed nodes of all the phases into a \emph{single} post-shattering instance.

The following theorem establishes \argOneLink{} for this situation.

\ifx\thetmpCounter\undefined
\newcounter{tmpCounter}
\fi
\setcounter{tmpCounter}{\thetheorem}
\setcounter{theorem}{\thesavedMultiPhase}

\addtocounter{savedMultiPhase}{1}
\edef\temp{theorem.\thesavedMultiPhase.formal}
\addtocounter{savedMultiPhase}{-1}
\hypertarget{\temp}{
\begin{restatable}[Formal]{theorem}{thmMultiphase}
\customlabel{thm:multiPhase}{formal}
    Let $p\in(0,1)$, $k \in \naturals^*$, and $S \subseteq V$ be a set of nodes. Consider a $k$-phase process where nodes in the set $S$ have some probability of "failing" each phase. Suppose that in each phase $i \in [k]$ and regardless of what occurred in previous phases (i.e., for an arbitrary conditioning), for each subset $U \subseteq S$ of nodes that have not "failed" in a phase $<i$, the probability that all nodes in $U$ "fail" phase $i$ is at most 
    \begin{align}p^{\card{U}}~. \label{condition:multiPhase}
    \end{align}
    Then the probability that each node in the set $S$ fails at least one phase is at most $(pk)^{\card{S}}$
\end{restatable}
}
\setcounter{theorem}{\thetmpCounter}

\subsection{Proof of \texorpdfstring{Theorem~\ref{thm:multiPhase}}{Theorem \ref*{thm:multiPhase}}}

\begin{proof}[Proof of Theorem~\ref{thm:multiPhase}]
When all nodes in a given set of nodes $S$ each fail at least one phase in a multi-phase process, each node $v$ has a phase that was the first one that it failed. Writing down for each node the first phase that it failed results a $\card{S}$-tuple of numbers in $[k]$. We call such a tuple the \emph{failing profile} of a failing execution. In the next paragraphs, we bound the probability that each node in $S$ fails at least one phase by analyzing the probability that all nodes in $S$ fail with a specific failing profile, and do a union bound over all failing profiles.

Let $l = \card{S}$ and let us index the nodes in $S$: $S=\{u_1,\ldots, u_l\}$. Consider an execution of the multi-phase process. For each node $u\in S$ and phase $i\in[k]$ introduce the event $\mathcal{E}_{u,i}$ that holds if $u$ did not fail in phases $<i$ but fails in phase $i$. 

Let $\mathcal{E}_{S}$ be the event that all nodes in $S$ fail, and for $x\in [k]^l$, let $\mathcal{E}_{S,x}$ be the event $\wedge_{j=1}^l \mathcal{E}_{u_j, x_j}$ that holds if the nodes in $S$ fail exactly in the phases specified by the profile $x$.
For a fixed profile $x\in [k]^l$ let $U_{x,i}=\{u_j\in S \mid x_j=i\}$ be the subset of nodes specified to fail in phase $i$ according to profile $x$. By repeated application of the lemma assumption with sets $U_{x,1}, \ldots, U_{x,k}$, we obtain $\Pr[\mathcal{E}_{x,S}]\leq p^{|S|}$. We obtain. 

\begin{align}
\Pr[\mathcal{E}_S]= \sum_{x\in [k]^l} \Pr[\mathcal{E}_{S,x}]\leq \sum_{x\in [k]^l]}p^{|S|}\leq (kp)^{|S|}~.
\end{align}
\end{proof}
One may wonder whether the union bound as in Theorem~\ref{thm:multiPhase} is necessary. The easiest way to see that one cannot bound the probability by $p^{|S|}$ is given by the special case where $S$ just consists of a single node that may fail in each of $k$ phases with probability $p$ if it did not fail before. Clearly, its probability to fail in one of the $k$ phases is
$1-(1-p)^k$. This approximately equals $p\cdot k$ when $p$ is small. 

Theorem~\ref{thm:multiPhase} can be used to show that the pre-shattering phases of various coloring algorithms do shatter into small components, e.g., \cite{CLP18,HKNT_stoc22}. These algorithms consists of $O(\log^*\Delta)$ phases, and in each phase nodes may not meet the pre-conditions of the next phase, e.g., too few of their neighbors might have gotten colored, resulting in the node failing and moving to the post-shattering phase. A similar proof was has been used in \cite{PS15} for establishing a shattering phenomenon for coloring triangle-free graphs.

\subsection{Ultra-Fast Shattering in Distributed Graph Coloring}

We re-establish the following result by Chang, Li, and Pettie.
\begin{theorem}[\cite{CLP20} redux]
$\Delta+1$-list coloring has complexity $O(\Det_{\mathsf{d1LC}}(\poly\log n))$. 
\label{thm:clp}
 \end{theorem}

The algorithm of \cite{CLP20} has a pre-shattering procedure of $O(\log^* \Delta)$ rounds, in each of which nodes may drop out of the process and proceed to a subsequent post-shattering phase. 
There are two immediate ways to attempt to re-instate a correct shattering analysis.
One is to construct $O(\log^* \Delta)$ post-shattering instances and to solve them sequentially.
The other is to use a single post-shattering instance for the whole process albeit with a larger component size, $O(\Delta^{\log^* \Delta} \log n)$. Neither of those results in the tight bound obtained by \cref{thm:clp}, which we instead re-instantiate via Theorem~\ref{thm:multiPhase}.

\begin{proof}[Proof sketch of \Cref{thm:clp}]
The coloring algorithm by \cite{CLP20} fits naturally into the multi-phase framework of Theorem~\ref{thm:multiPhase} because its pre-shattering part consists of  a sequence of $O(\log^* \Delta)$ local randomized steps running for $O(1)$ rounds. Each step may add some nodes to post-shattering to ensure that some properties hold in later steps, but the probability of joining post-shattering never exceeds $1/\poly(\Delta)$ at any given node in any step (in fact, it is often $\exp(-\Omega(\Delta))$). See Lemmas 2.5, 4.2 to 4.5, and 5.1 in~\cite{CLP20} for the analysis of the steps that can add nodes to post-shattering.
A vertex is sent to the residual instance as soon as it fails in one of the steps, so the final residual set is the union of failures over all steps rather than the failure set of a single local experiment. Theorem~\ref{thm:multiPhase} is designed exactly for this situation: if, in every phase and under arbitrary conditioning on previous phases, any still-active well-separated set fails that phase with small enough probability with respect to the degree and the number of phases, then even the union of all phase failures shatters. In short, Theorem~\ref{thm:multiPhase} can be applied to their algorithm, and \Cref{cor:shattering-small-cc} then turns the resulting separated-set bound into small connected residual components, establishing their main claim. 
\end{proof}
The shattering arguments of the more recent algorithms for $\deg+1$-list coloring can be re-established in a similar manner, in both LOCAL~\cite{HKNT_stoc22} and CONGEST~\cite{HNT22}.
These guarantees have already been partially adapted in the respective journal versions, based on material from this manuscript.

\section*{Acknowledgements}
We have discussed the findings in this work with countless people and are thankful for their feedback. In particular, we would like to thank Sebastian Brandt for his contributions towards finding \Cref{example:FG}, and Christoph Grunau as well as Václav Rozhoň for helpful discussions on \Cref{sec:fg-lll}. This research was funded in whole or in part by the Austrian Science Fund (FWF) \url{https://doi.org/10.55776/P36280}, \url{https://doi.org/10.55776/I6915}. For open access purposes, the author has applied a CC BY public copyright license to any author-accepted manuscript version arising from this submission.

\bibliography{0a_refs_preamble,0b_refs_shattering}
\bibliographystyle{plain}
\appendix

\section{Small Ball Covers in Post-shattering}
\label{sec:ballCover}
\paragraph{Degree Reduction.}
Shattering-based vertex coloring algorithms, e.g., \cite{BEPS16_jacm,CLP18,HKMT21,HKNT_stoc22}, include a degree reduction phase that is not part of the shattering approach itself. After this phase the graph induced by uncolored vertices has maximum degree $\Delta'\leq \poly\log n$. Only then, the pre-shattering phase is started and \Cref{cor:shattering-small-cc} is applied with $\Delta'$. This ensures that connected components are of size $\poly(\Delta')\log n\leq \poly\log n$. As a result of this exponential decrease in instance size, the post-shattering phase runs extremely fast. 

\paragraph{Efficient algorithms without degree reduction.}
For problems like the maximal independent set problem or the Lovász Local Lemma problem such a degree reduction procedure is not known. Thus the size of the remaining components can only be bounded by $\poly\Delta\log n\gg \poly\log n$ whenever the maximum degree $\Delta$ is large \cite{BEPS16_jacm,Ghaffari_soda16}. To circumvent the problem the MIS algorithms of \cite{BEPS16_jacm} employs a second pre-shattering phase that is run on the graph induced by the nodes remaining undecided after the first pre-shattering phase. 
Then, they compute a so-called well spaced out ruling set of the still undecided nodes and form a virtual \emph{cluster graph} by letting each undecided node join the cluster of the closest ruling set center and connecting two clusters if they are connected by an edge. By leveraging the technical details of the employed ruling set algorithm and the guarantees of both shattering phases, they show that each connected component of this virtual cluster graph only has a logarithmic number of clusters. This is exploited down the line for the design of a $\poly\log\log n$-round post-shattering phase, see \cite{BEPS16_jacm,Ghaffari_soda16} or the paragraph below \Cref{cor:shattering-ball-graph} for more details. The usage of two pre-shattering phases works quite well for greedy-type problems like MIS but would be significantly more technical or possibly even impossible to exist for non-greedy type problems like LLL. Thus we present a standalone result of the same type that works for any pre-shattering procedure meeting the requirements of \Cref{lem:ShatteringGeneral}. We begin with defining the necessary cluster graphs. 

\paragraph{A note on notation.} At multiple points in the upcoming definition and corollary, we consider an induced subgraph $G^c[B]$, where $c$ is a positive integer, $G=(V,E)$ is a graph, and $B \subseteq V$ is a subset of the nodes.
$G^c[B]$ is an induced subgraph of the power graph $G^c$, that is, it is the graph obtained by first taking the power $G^c$ of the graph $G$, and then considering the subgraph induced by $B$.
$G^c[B]$ should not be confused with $(G[B])^c$, which is the graph one would obtain by first considering the induced subgraph $G[B]$ of $G$, and then taking a power graph of this induced subgraph.

\begin{definition}[$(c,r)$-cluster graph]
\label{def:ballGraph}
    Let $c$ and $r$ be positive integers.
    Consider a graph $G = (V,E)$, with subsets $R \subseteq B \subseteq V$ of nodes.
    Then, the $(c,r)$-cluster graph of $R$ over $B$ is an assignment $b:B \to \set{R,\bot}$ and a graph $G' = (R,E')$ defined as follows:
    \begin{enumerate}
        \item For each $v \in B$, $b(v) \in R$ is the \emph{ball} of $v$. It is the element of $R$ nearest to $v$ in $G^c[B]$, with ties broken with IDs. If $\dist_{G^c[B]}(v,R) = \infty$, then $v$ is part of no ball, which we denote by $b(v) = \bot$.

        For an element $u \in R$, we denote by $\Ball(u) = b^{-1}(u) \subseteq B$ the nodes $v \in B$ s.t.\ $b(v) = u$. 
        \item For each $u,u' \in R$, there is an edge between $u$ and $u'$ in $G'$ iff $\dist_G(\Ball(u),\Ball(u')) \leq r$, i.e., if there exists $(v,v') \in \Ball(u) \times \Ball(u')$ s.t.\ $\dist_G(v,v') \leq r$.
    \end{enumerate}
\end{definition}

Note that given a sets $R$ and $B$ the cost of computing a cluster graph essentially scales with how far a node in $B$ can be from the nearest node in $R$. If $\max_{v\in B} \min_{w \in R} \dist_{G^c}(v,w) \leq z$, then a $B$-covering $(c,r)$-cluster graph of $R$ only takes $O(c\cdot z+r)$ rounds of LOCAL to compute. 

\begin{corollary}[Post-shattering connected components admit small ball covers]\label{cor:shattering-ball-graph}
Let $r,c_2$ be positive integers, let $c_3 \geq 0$, and consider a process that generates a random subset $B \subseteq V$ such that 
\begin{align} \label{condition:THESET3}
    \Pr[S\subseteq B]\leq \Delta^{-(2c_2+\max\{c_2,r\}+c_3+2) |S|}
\end{align}
holds, for each $c_2$-independent set $S\subseteq V$ of size $t\geq \log_\Delta n$. Then the following holds with probability at least $1-\Delta^{-c_3t}$:
    For every independent set $R$ of $G^{c_2}[B]$, every connected component of the $(c_2,r)$-cluster graph of $R$ over $B$ has at most $t$ nodes.
\end{corollary}

\begin{proof}
Let $x=2c_2+\max\{c_2,r\}$, let $R$ be an independent set of $G^{c_2}[B]$, let $\cC$ be a connected component in the cluster graph of $R$, and let $R_{\cC}\subseteq R$ be the nodes of $R$ in that connected component.
We describe a process to extend $R_{\cC}$ to a superset $R' \subseteq B$ that is $c_2$-independent $x$-connected.
Because such a set $R'$ must contain no more than $t$ nodes by \cref{lem:ShatteringGeneral} with probability at least $1-\Delta^{-c_3t}$, the set $R_{\cC}\subseteq R'$ must be similarly bounded in size, proving the claim. 

Let us initialize $R' \gets R_{\cC}$, and let $u_{\max}$ be the element of $R_\cC$ of highest ID. Throughout the proof, let us consider the connected component of $u_{\max}$ in the graph $G^x[R']$, $\CC_{G^x[R']}(u_{\max})$, which we denote $\CC(u_{\max})$ for short.
Note that the set of nodes in $\CC(u_{\max})$ is $x$-connected by definition.
$R'$ is initially $c_2$-independent, being a subset of $R$, and our goal is to grow $R'$ with nodes from $B \setminus R$ until it is $x$-connected while keeping it $c_2$-independent.  We grow $R'$ such that $R'\setminus R_{\cC} \subseteq \CC(u_{\max})$, so that any node not yet in $\CC(u_{\max})$ is in $R_{\cC}$.
This means that the set $R'$ will be $x$-connected once $R_{\cC} \subseteq \CC(u_{\max})$.

While $R_{\cC} \not \subseteq \CC(u_{\max})$, consider $R_{\cC} \cap \CC(u_{\max})$ and $R_{\cC} \setminus \CC(u_{\max})$. $R_{\cC}$ is connected in the cluster graph, therefore, there exists some $u \in R_{\cC} \cap \CC(u_{\max})$ and $u' \in R_{\cC} \setminus \CC(u_{\max})$ such that $uu'$ is an edge from the cluster graph. This means that there exist some $v \in \Ball(u)$ and $v' \in \Ball(u')$ s.t.\ $\dist_G(v,v') \leq r$. Consider two such nodes, as well as the shortest path $u_0 = u,\ldots,u_k = v$ and $u_{k+1} = v',\ldots,u_{k+k'} = u'$ in $G^{c_2}[B]$.

We now greedily add nodes from the path $u_0,\ldots,u_{k+k'}$ to $R'$ until $u'$ is in $\CC(u_{\max})$. Let us consider the first node $u_i \not \in R'$ on the path at distance $> c_2$ from all nodes in $\CC(u_{\max})$. It must exist since the last node on the path has distance $> x$ from $\CC(u_{\max})$. Since $u_i$ is the first such node, it is at distance $\leq \dist_G(u_{i-1},\CC(u_{\max})) + \dist_G(u_i,u_{i-1}) \leq c_2 + \max(c_2,r)$ from a node in $\CC(u_{\max})$. $u_i$ therefore cannot be within distance $\leq x-c_2 - \max(c_2,r)=c_2$ from a node $w \in R'$ in another connected component of $G^x[R']$ than $\CC(u_{\max})$, as otherwise, $w$ would be within distance $x$ from $\CC(u_{\max})$, a contradiction. 
Thus, we can add $u_i$ to $R'$ and the connected component $\CC(u_{\max})$ has grown. We can iterate the process with the next vertex on the path with distance $> c_2$ from all nodes in $\CC(u_{\max})$ until $u'$ is part of $\CC(u_{\max})$, and then repeat the process for another $u' \in R_{\cC} \setminus \CC(u_{\max})$ until $R_{\cC} \subseteq \CC(u_{\max})$.

\begin{figure}
\centering
\includegraphics[width=0.8\textwidth]{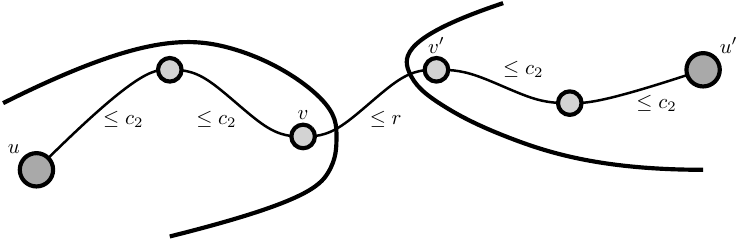}
\caption{
When two ball centers $u$ and $u'$ are connected in the $(c_2,r)$-cluster graph, there exist two nodes $v \in \Ball(u)$ and $v' \in \Ball(u')$ in their respective balls at distance $\leq r$ from one another in $G$. Furthermore, there exists a path in $G^{c_2}[B] \cap \Ball(u)$ between $u$ and $v$ (and similarly with $u'$ and $v'$).
}
\label{fig:rulingSet}
\end{figure}
\end{proof}

A typical use of \cref{cor:shattering-ball-graph} involves computing a ruling set after the pre-shattering phase, and to compute the cluster graph of this ruling set.
Consider a pre-shattering phase in which no $c_2$-independent set of size $\geq \log_\Delta n$ joins post-shattering with high probability. Let us compute a $(2,O(\log \log n))$-ruling set $R$ of $G^{c_2}[B]$, the subgraph of $G^{c_2}$ induced by $B$, the set of nodes in the post-shattering instance: the task admits a $O(c_2 \log \log n)$-round randomized algorithm that is correct with high probability in $n$~\cite{GV_RulingSetBoundedGrowth_podc07,SEW13}. Computing the $(c_2,r)$-cluster graph of $R$ only takes $O(c_2 \log \log n + r)$ rounds, as each node in $B$ is at distance at most $O(\log \log n)$ in $G^{c_2}[B]$ from the nearest ball center. Assuming that we are solving a problem that is greedy in nature, we can extend the partial solution produced by the pre-shattering phase to nodes in post-shattering by the generic technique of \emph{network decomposition}. Computing the network decomposition can be done in $\Otilde((c_2 \log \log n +r)\log^2 \log n)$ rounds~\cite{GG24-FOCS}, where the $O(c_2 \log \log n + r)$ factor comes from the cost of simulating a round of communication over the cluster graph. This allows, e.g., to solve MIS in $\Otilde(\log^3 \log n)$ rounds. The parameter $r$ serves to add some distance between clusters. This can be useful, e.g., for problems where solving a node in post-shattering may require making some changes around that node, including.

\Cref{cor:shattering-ball-graph} is inspired by similar results in \cite{GHK18,MPU23} and simplifies their process and its analysis in a standalone manner that is independent of the problem and pre-shattering process at hand. In fact, it shows that already after a single pre-shattering phase and a suitable ruling set computation, the number of balls in each connected component of the cluster graph is small. This can then be leveraged for the design of an efficient post-shattering phase in the same way as done in \cite{BEPS16_jacm,MPU23}.

\section{Shattering with Dependent Opportunities: Applications}
\label{app:OpportunitiesApplications}
\subsection{Application I: \texorpdfstring{$\deg+1$}{deg+1}-list coloring  (Section 3 in 
\texorpdfstring{\cite{BEPS16_jacm}}{[BEPS16, JACM]})}
\label{sec:coloring-beps}
 In an instance of the $\deg+1$-list coloring problem each node of a graph has a list of colors whose size exceeds its degree. The objective is to assign a color to each node from its list such that adjacent vertices receive different colors. In general there is no assumption on the color space. In the $\Delta+1$-coloring vertex problem each node has the same list of size $\Delta+1$. For illustrative purposes we prove the following theorem.
\begin{theorem}[\cite{BEPS16_jacm}]
\label{thm:BEPScoloring}
Let $c_3\geq 1$. There is a $O(\log \Delta)+\poly\log\log n$-round randomized algorithm that solves the $\deg+1$-list coloring problem on any $n$-node graph with probability at least $1-n^{-c_3}$. 
\end{theorem}
The runtime of \Cref{thm:BEPScoloring} is outdated and subsumed by faster but more complicated algorithms \cite{CLP18,HKMT21,HN_tcs23,HKNT_stoc22,HNT22}.

The heart of the shattering argument for the $\deg+1$-list coloring problem is the following one-round procedure that is iterated $O(\log \Delta)$ times.

\begin{algorithm}[ht]
    \caption{\OneShotColoring, for node $v$}
    
    $v$ picks a color $c_v$ u.a.r.\ in its palette $\Psi_v$, sends $c_v$ to its neighbors $N(v)$.
    
    $v$ makes $c_v$ its permanent color if $\forall u \in N(v), c_u \neq c_v$. 
    
    \label[algorithm]{alg:OneShotColoring}
\end{algorithm}

\begin{lemma}[Subset of Lemma 5.2 in \cite{BEPS16_jacm}]
\label{lem:johasson}
    $\Pr[v \text{ is colored in }\OneShotColoring] \geq 1/4$. \cite{Johansson99}
\end{lemma}

\begin{proof}[Proof Sketch of \Cref{thm:BEPScoloring}]
In this proof sketch we omit the details of the $O(\log \Delta)$-round procedure that with high probability reduces the maximum degree of uncolored instances to $\hDelta=\poly\log n$. In fact, their degree reduction  constructs a constant number of $\deg+1$-list coloring instances on graphs with polylogarithmic maximum degree that need to be solved sequentially. We next explain how to efficiently handle each of these. In  particular, we use \Cref{thm:probability-guarantee} to reprove the central shattering claim in \cite[Lemma 5.3]{BEPS16_jacm}. 

Consider $T=(c_3+5)\log_{4/3}\hDelta$ iterations of \OneShotColoring and a $2$-independent set $S$. By \Cref{lem:johasson} each node in $S$ has the opportunity to get colored (be decided) in every round; the probability to miss a single opportunity is at most $3/4$  and the success of opportunities is independent for different nodes in $S$. Thus, in total there are $T|S|$ opportunities and by \Cref{thm:probability-guarantee} the probability that all nodes in $S$ remain undecided is at most $p^{T|S|}\leq \Delta^{-(5+c_3)|S|}$. 

By \Cref{cor:shattering-small-cc} the graph into independent components of size at most $N=\poly \hDelta\log_{\Delta}n\leq \poly\log n$ with high probability. Each of these can be solved via a deterministic algorithm in time $\poly\log N$, e.g., in $O(\log^3 N)=O(\log^3\log n)$ rounds with the algorithm from \cite{GK21}, or polynomially faster with recent faster but more involved deterministic algorithms \cite{GG23,GG24-FOCS}. 
\end{proof}

\subsection{Application II: Ghaffari's seminal MIS algorithm \texorpdfstring{\cite{Ghaffari_soda16}}{}}
\label{sec:mis-ghaffari}

In this section, we present a proof of the following theorem of \cite{Ghaffari_soda16} with some necessary adjustments to the original proof of \cite{Ghaffari_soda16}. For a discussion on how this proof compares to the one in \cite{Ghaffari_soda16}, please refer to the end of this section.

\begin{restatable}{theorem}{thmMIS}[\cite{Ghaffari_soda16,GG24-FOCS} redux]
 \label{thm:MIS}
There is a randomized distributed algorithm that w.h.p.\ computes a maximal independent set in $O(\log\Delta)+\Otilde(\log^{3}\log n)$ rounds. 
 \end{restatable}

We first present Ghaffari's seminal MIS algorithm \cite{Ghaffari_soda16,JSX_dc16}.

\begin{algorithm}[ht]
    
    Initialize $I \gets \emptyset$ and $C \gets \emptyset$.
        
    \nonl Throughout the algorithm, $I$ is the current set of MIS nodes, and $C$ their neighbors. 
    
    Each $v\in V$ is initially active and sets its local probability of joining the MIS to $p_{0}(v) = 1/2$.
    
    \For{ $t = 0,\dots,T-1$, at each active node $v$, }{
        Compute $d_{t}(v) = \sum_{u \in V \setminus C} p_{t}(u)$.
        
        $v$ marks itself with probability $p_{t}(v)$.
        
        \If{$v$ is marked and $N(v)$ contains no marked node}{$v$ joins $I$. Stop the algorithm for $v$.}
         
        \If{a neighbor of $v$ joined $I$}{$v$ joins $C$, stop the algorithm for $v$.}

        \If{$d_{t}(v) \geq 2$}{let $p_{t+1}(v) \gets p_t(v)/2$}
    
        \Else({($d_{t}(v) < 2$)}){let $p_{t+1}(v) \gets \min(2p_t(v),1/2)$}
    }
       
    \caption{MIS pre-shattering of Ghaffari's algorithm \cite{Ghaffari_soda16,JSX_dc16}, running for $T$ rounds}
    \label[algorithm]{alg:mis-shattering}
\end{algorithm}

We continue with the definition of golden rounds as introduced in \cite{Ghaffari_soda16}. This definition plays a major role in the analysis of the algorithm.

\begin{definition}
\label{def:golden-rounds}
    We say a node is \emph{low-degree} if $d_t(u) < 2$, and $\emph{high-degree}$ otherwise. We say $u$ is in a \emph{golden round} if one of two situations arise: 
    \begin{itemize}
        \item\emph{Golden round of type 1:} if $d_t(u)<2$ ($u$ is low-degree) and $p_t(u)=1/2$.
        \item\emph{Golden round of type 2:} if $d_t(u) \geq 1$ and $\card{d_t^{-1}([0,2]) \cap N(u)\setminus C} \geq d_t(u) / 10$.
    \end{itemize}
\end{definition}

\begin{lemma}[{\cite[Lemma 3.1]{Ghaffari_soda16}}]
    \label{lem:golden-rounds}
    Let $\beta \geq 9$. After $T=\beta(\log \deg(v) + \log 1/\eps)$ rounds of Algorithm \ref{alg:mis-shattering}, a node $v$ either joined $I$ or $C$, or went through $\frac{\beta}{13}(\log \deg(v) + \log 1/\eps)$ golden rounds.
\end{lemma}

Note that the two types do not a priori exclude each other. In fact, it is not too difficult to construct an example in which a node is in a golden round of both types: consider the initial round of the algorithm in a $3$-regular graph of girth $>2$, each node is of low-degree $d_0(u) = 3/2$, has marking probability $p_0(u) = 1/2$, and has only neighbors of low-degree.

\begin{lemma}[{\cite[Lemma 3.2]{Ghaffari_soda16}}]
\label{lem:goldenRoundProbability}
    Consider a node $v$ at some round $j$. If $j$ is a golden round of type $1$ for $v$, $v$ joins the MIS with probability at least $1/200$. If $j$ is a golden round of type $2$ for $v$, then a neighbor of $v$ joins the MIS with probability at least $1/200$. 
\end{lemma}
The original Lemma 3.2 in \cite{Ghaffari_soda16} further claims that the events depend solely on the $1$-hop and $2$-hop neighborhoods, respectively. However, as discussed at the end of this section, this holds true only within a single iteration of the algorithm, assuming the prior randomness is fixed. Thus it cannot be combined with \Cref{lem:Shattering} to show that the process shatters.

The main contribution of this section is the proof of the following lemma showing that the probability for an $4$-independent set of node to participate in the post-shattering phase is small. It replaces the faulty counterpart in \cite{Ghaffari_soda16}, see the end of this section for more details on the analysis in the original paper. The rest of the algorithm and analysis remains unchanged. 

\begin{lemma}
\label{lem:MISSetProbability}
Let $S$ be a $4$-independent set of size $|S| \geq 10^5$ and $c_3 \geq 0$. Then after $T> 13(14+c_3) \log_{200/199} \Delta$ iterations of \Cref{alg:mis-shattering} (Ghaffari's MIS algorithm), the probability that all nodes of $S$ are undecided is at most $(199/200)^{|S|\cdot T/13}\leq \Delta^{-(14+c_3)|S|}$. 
\end{lemma}
\begin{proof}
Consider a $4$-independent set $S$. If we modify the definition of a golden round such that any round is golden for a vertex once it is decided \Cref{lem:golden-rounds} implies that the nodes collectively experience $|S|\cdot T/13$ golden rounds during the $T$ iterations of the algorithm. Due to \Cref{lem:goldenRoundProbability} a node fails to be decided with probability $199/200$ in a golden round. Interpreting golden rounds as an opportunity, \Cref{thm:probability-guarantee} implies that the probability that all nodes in $S$ are undecided at the end of the $T$ rounds is at most $p^{|S|\cdot T/13}\leq \Delta^{(14+c_3)|S|}$~. 
\end{proof}

We next provide a proof sketch of \Cref{thm:MIS}, which given the proof of \Cref{lem:MISSetProbability} and the general shattering lemma is along the same lines as in \cite{Ghaffari_soda16}.
\begin{proof}[Proof Sketch of \Cref{thm:MIS}]
Let $B$ be the set of nodes that are undecided after $O(\log\Delta)$ rounds of \Cref{alg:mis-shattering}. 
By \Cref{lem:MISSetProbability} any $4$-independent set of nodes $S\subseteq V$ is contained in $B$ with probability at most $\Delta^{(14+c_3)|S|}$.  Use the algorithm of \cite{SEW13} to compute a $(2,z)$ ruling set of $G^5[B]$ with $z=O(\log\log n)$ in $O(\log\log n)$ rounds. Then, form the cluster graph according to \Cref{def:ballGraph} with $r=1$ and $c=3$. 
By \Cref{lem:ShatteringGeneral,cor:shattering-ball-graph} any connected component of the cluster graph has at most $N=\log_{\Delta}n$ nodes. Compute a network decomposition with diameter $d=O(\log N)=O(\log \log n)$ and $c=O(\log N)=O(\log\log n)$ colors classes via the algorithm of \cite{GG24-FOCS} in $z\cdot O(\log^2 N \poly\log\log N=\Otilde(\log^3 \log n)$ rounds. 
Then lift the network decomposition back to a network decomposition of $G[B]$ with the same parameters, except that the cluster diameter bound increases by a factor $5z$. Finally iterate through the color classes of the network decomposition and solve the MIS problem on all clusters with the same color in parallel (taking choices of already processed neighboring clusters into account). The runtime of this last step is $\mathsf{clusterDiameter}\cdot \mathsf{\#colorClasses}=O(\log\log n\cdot O(z\log\log n)=O(\log^3\log n)$.  The runtime is dominated by the computation of the network decomposition. 
\end{proof}

\paragraph{The algorithm/analysis in \cite{Ghaffari19}.}
The abstract and \Cref{lem:goldenRoundProbability} contain the following sentence that forms the basis of their shattering argumentation: \emph{"The guarantee holds even if the randomness outside 2-hops neighborhood of $v$ is determined adversarially"}. Intuitively, one may believe that hence the events of being decided for the nodes in any $5$-independent set are independent and hence the argument should immediately provide a shattering phenomenon. But this intuition is misleading. In \Cref{sec:counterExample}, we provided an counterexample where a process satisfying such a property does not shatter; simply put the used respective shattering lemma (\Cref{lem:shatteringWrong}) is not correct. One way out would be to use the correct shattering lemma from \cite{FG17}, see \Cref{lem:Shattering}. We next illustrate that this clearly cannot be done by introducing a different view on the algorithm. Throughout the execution of the algorithm all $p_t(v)$ values of the nodes are in the set $\cP=\{1/2^i,1\leq i\leq O(\log\Delta)\}$. We can reformulate the algorithm in a way that every every node in every round flips $|\cP|$ coins, that is, one coin for each  value in $\cP$. The outcome of the algorithm is identical, as a node $v$ discards most of these coins flips and acts only according to the value of the coin corresponding to its actual $p_t(v)$ value. Now the statement that the status of a node only depends on the outcome of the coins in its $O(1)$-hop neighborhood (as it would be required for \Cref{lem:Shattering}) is clearly wrong. The outcome of a node in round $i$ is influenced by the coin flips in its $i-1$-neighborhood. In fact, the outcome of the coin flips in previous rounds determine the $p_t(v)$ values and hence they determine which coins nodes use/discard in a specific round. Hence, despite the fact that in a golden round a node has a constant probability to be decided regardless of the randomness outside its $2$-hop neighborhood chosen by an oblivious adversary,  the events whether a node is decided during the whole execution is not guaranteed to be independent for a set of nodes that is only $O(1)$-independent. It is unlikely that the same statement as in \Cref{lem:goldenRoundProbability} would be true for an adaptive adversary that can see all $|P|$ coins flips of all nodes inside the $2$-hop neighborhood.

\subsection{Application III: Earlier MIS Algorithm (Section 3 in 
\texorpdfstring{\cite{BEPS16_jacm}}{[BEPS16, JACM]})}
\label{sec:mis-beps}
In this result, failure corresponds to neither being in the independent set nor being adjacent to it. Let $I$ be the independent set constructed by their procedure, and $\hGamma(I) = I \cup N(I)$. $\hGamma(I)$ is the set of \emph{unsolved} nodes, a node being \emph{solved} whenever it or one of its neighbors joins the MIS.

The authors emphasize that the events that nodes join $\hGamma(I)$ are independent for nodes at a distance at least $5$ from each other \emph{in any given round}. Still, similar to Ghaffari's algorithm the reasoning about the survival probability of any $4$-independent set is more tricky. The reason for this is, in contrast to the coloring algorithm from \Cref{sec:coloring-beps}, that nodes do not automatically have a good probability to be decided in every round. Instead, similar to Ghaffari's algorithm, the authors show that each node has some $c\log \Delta$ iterations in each of which it is decided with constant probability, but there are dependencies between distant nodes about when this is going to happen. Here, the constant $c$ can be arbitrarily chosen, but it influences the constant in the $O$-notation of the algorithms' runtime, see \Cref{alg:BEPSMis}.

Still, using the language of opportunity decision processes we obtain that there are $|S|\cdot c\cdot \log \Delta$  opportunities in total for the nodes in a $4$-independent set of nodes $S$.  By \Cref{thm:probability-guarantee} the probability that each node in $S$ is undecided after running the algorithm is at most $p^{c\log \Delta|S|}\leq \Delta^{-(14+c_3)|S|}$, using that $c$ is a sufficiently large constant. This is sufficient to apply \Cref{lem:ShatteringGeneral,cor:shattering-ball-graph} to show that the process shatters the graph.

\paragraph{The analysis in \cite{BEPS16_jacm}.} We now explain the reasoning in  \cite{BEPS16_jacm} in more detail. In scale $k$, where all the nodes have degree at most $\Delta/2^{k-1}$, in each Luby step, a node of degree $>\Delta/2^k$ joins $\hGamma(I)$ w.p.\ at least $(1-e^{-1/2})e^{-1}>1/7$ (Lemma 3.1 in \cite{BEPS16_jacm}).  A node that is in the bad set $B$ at the end of the algorithm had degree above $\Delta/2^k$ throughout the scale $k$ at the end of which it was added to $B$. Therefore, each node has a probability of surviving that is at most $(6/7)^\ell$ where $\ell = c \log \Delta$ is the length of each scale, which is $< \Delta^{-c/5}$ as $(6/7)^5 < 1/2$. Additionally, for a $4$-independent set of nodes $S$ we have independence of the events that they get solved in any given Luby round, when conditioned on arbitrary behavior in previous rounds. The authors claim that this property implies that such a set $S$ stays unsolved with probability at most $\Delta^{-c' \card{S}}$ with $c$ taken so that $c \geq 5c'$, and with $c'$ a large enough constant, the algorithm shatters. An easy way to show this despite the fact that opportunities of different nodes in different rounds depend on each other is to use Theorem~\ref{thm:opportunityProcess}; here an opportunity for a node is \emph{successful} if it gets solved.

Another minor issue in the proof in \cite{BEPS16_jacm} is that nodes of higher degree should have priority over lower degree nodes when it comes to joining the MIS (see proof of Lemma 3.1). Otherwise, one does not get that a node which tries to join the MIS has a constant probability to join it. Indeed, consider a node $v$ of high degree, adjacent to many of small degree, which all try to join the MIS with much higher probability than $v$. If they can prevent $v$ from joining the MIS, $v$ does not have a good chance of success when making an attempt to join: it could be as low as $2^{-\deg(v)}$ if $v$ only has neighbors of degree $1$.

\subsection{Application IV: Maximal Matching (Section 4 in 
\texorpdfstring{\cite{BEPS16_jacm}}{[BEPS16, JACM]})}
\label{sec:mm-beps}
Their maximal matching algorithm of \cite{BEPS16_jacm} first uses $O(\log \Delta)$ rounds to reduce the maximum degree of the unsolved part to $\Delta'=\poly\log n$ w.h.p. Afterwards they use a classic shattering procedure yielding unsolved components of size $\poly\Delta'\log n=\poly\log n$ that can be solved independently in post-shattering.
We focus on their second contribution here, their shattering argument. The crux of the argument is the following lemma about their main procedure, inspired by an earlier algorithm by Israeli and Itai \cite{II86}.

\begin{lemma}[Lemma 4.2 in \cite{BEPS16_jacm}]
    \label{lem:mm-degree-drop}
    For any integer $i$ let $\deg_i(v)$ be the degree of $v$ in iteration $i$ of \cref{alg:mm-beps}. Then in any round $i$
    \[
    \Pr[\deg_{i+1}(v) \leq \frac{3}{4}\deg_i(v) ] \geq \frac{1}{4}
    \ .\]
\end{lemma}

Furthermore, conditioned on previous rounds, whether a node is solved in any given round of their procedure is only dependent on random choices within distance $3$ of the node\footnote{They actually use a procedure with dependencies at higher distance, but only half of it is relevant for \cref{lem:mm-degree-drop}, which is the part presented in \cref{alg:mm-beps}}.

They say that a stage (iteration) of the algorithm is a \emph{successful stage} for a node $v$ if its degree reduces by at least a $3/4$ factor in that iteration, which \cref{lem:mm-degree-drop} shows occurs w.p.\ at least $1/4$. Observe that any stage is successful for a node once all its edges incident are decided.
Let $X_v$ be the random variable counting the number of successful stages of node $v$. As each stage is successful with probability at least $1/4$, for a $T$-round procedure the random variable $X_v$ is distributed like the number of heads when throwing $T$ coins that show head with probability $1/4$. But these random variables are not independent for different nodes. Still, by \Cref{thm:probability-guarantee} applied with $t_v=T/8$ we obtain  that the probability that all nodes in a $5$-independent set $S$ see less than $T/8$ successful stages is at most
\begin{align}
     \prod_{v\in S} \Pr[s_{T}\leq t_v]\leq e^{-T/12 |S|}=\Delta^{-(8+c_3)|S|}. 
 \end{align}
 Here $s_{x_v}$ is the random variable representing the number of successful outcomes (heads) when flipping $x_v$ coins with the probability of obtaining head being $1/4$. The last inequality holds when $T\geq 12\cdot (8+c_3)\log \Delta$. This is sufficient to apply \Cref{lem:ShatteringGeneral,cor:shattering-small-cc} to prove that their algorithm shatters the graph.

\paragraph{The analysis in \cite{BEPS16_jacm}} For a $5$-independent set $S$ let $X=\sum_{v\in S}X_v$ denote the collective number of successful stages of its nodes.  By linearity of expectation $E[X]\geq |S|/4 \cdot T$ if we run $T$ iterations of the algorithm. One can show  that if $X$ is close to its expectation and for a sufficiently large $T=O(\log_{4/3} \Delta')$ then deterministically at least one of the nodes in $S$ has experienced enough successful stages to not participate in the post-shattering phase. The authors proceed to show that $X$ is well-concentrated around its expectation bounding the probability that all of $S$ participates in the post-shattering phase. All these definitions appear in the proof of Lemma~4.3 \cite{BEPS16_jacm}. Unfortunately, the reasoning that $X$ is well-concentrated around its expectation is faulty. They apply a Chernoff bound for variables with negative correlation indicating that the variables $X_v, v\in S$ are negatively correlated.
However, no argument is provided for this negative correlation.

To establish a shattering argument one can use \Cref{thm:opportunityProcessGeneralized} to bound the probability that an independent set of nodes participates in the post-shattering phase despite potential long-range dependencies between occurring opportunities. Here, an opportunity  is \emph{successful} for a node if its degree goes down by a constant multiplicative factor. Every single round is an opportunity in their algorithm. Observe that once a node is solved, its degree is $0$, so it keeps decreasing by a constant factor in each subsequent round and each additional round is another successful opportunity for that node; a node with $\log \Delta$  or more successful opportunities is necessarily solved.

\section{Explicit Examples of Non-Independence in \texorpdfstring{\cite{FG17}}{Fischer--Ghaffari}}
\label{app:examples}

\begin{restatable}{example}{LLLFreezingExample}
\label{ex:lll-nonindependent-freezing}
\label{example:FG}
Consider a cycle where nodes are indexed from $v_1$ to $v_n$, and edges are similarly labeled from $e_1$ to $e_n$, such that two edges incident to the node $v_i$, $i \in [n]$, are the edges $e_i$ and $e_{i+1}$ (with indices taken $\mod n$. Let each edge sample $4$ bits, and let $S_i$ be the random variable of the bits sampled by edge $e_i$. Let the bad event at node $i$ be $(S_i = 0000 \wedge S_{i+1} \neq 0000) \vee (S_{i+1} = 0000 \wedge S_{i} \neq 0000)$.
Let $p=2\cdot \frac{1}{16}\cdot \frac{15}{16}$ be the probability of each bad event.
We consider a progressive sampling process where the order in which events sample their random variables is from $v_1$ to $v_n$, and vertex $v_i$ samples the variables at $e_i$ before $e_{i+1}$.
Equivalently, the edges are considered from $e_1$ to $e_n$ to have their bits sampled. If the probability of a bad event exceeds $\sqrt{p}$ given the current sampling, then all the unsampled random variables of that bad event are frozen, i.e., they will not be sampled later in the sampling process. Let $F_j$ be the event that event $j$ has one or more frozen variables, then for any $i \in [3,n-1]$,
\begin{align}
    \Pr[F_1 \wedge F_i] 
    & = \frac{115}{1536} + \parens*{-\frac{1}{8}}^{i-1} \cdot \frac{49}{12}
    &
    \Pr[\neg F_1 \wedge F_i] & = \frac{1127}{4608} - \parens*{-\frac{1}{8}}^{i-1} \cdot \frac{2401}{576}
    \\
    \Pr[F_1 \wedge \neg F_i] & = \frac{245}{1536} - \parens*{-\frac{1}{8}}^{i-1} \cdot \frac{49}{12}
    &
    \Pr[\neg F_1 \wedge \neg F_i] & = \frac{2401}{4608} + \parens*{-\frac{1}{8}}^{i-1} \cdot \frac{2401}{576}
\end{align}
Notably, $\Pr[F_1 \wedge F_i] \neq \Pr[F_1] \cdot \Pr [F_i]$.
\end{restatable}

\begin{figure}[h]
\centering
\begin{tikzpicture}[every node/.style={circle, draw, inner sep=1pt, minimum size=6mm}]
\node (0) at (0,0) {$v_n$};
  \node (1) at (1.5,0) {$v_1$};
  \node (2) at (3,0) {$v_2$};
  \node (3) at (4.5,0) {$v_3$};
  \node (4) at (6,0) {$v_4$};
  \node (5) at (7.5,0) {$v_5$};
\draw[dashed] (-1,0) -- (0);
  \draw[dashed] (5) -- (8.5,0);
\draw (0) -- node[above=3pt, draw=none, fill=none] {$e_1$} (1);
  \draw (1) -- node[above=3pt, draw=none, fill=none] {$e_2$} (2);
  \draw (2) -- node[above=3pt, draw=none, fill=none] {$e_3$} (3);
  \draw (3) -- node[above=3pt, draw=none, fill=none] {$e_4$} (4);
  \draw (4) -- node[above=3pt, draw=none, fill=none] {$e_5$} (5);
\end{tikzpicture}
\label{fig:counterExample}
\caption{Labeling of the vertices and edges in the counterexample.}
\end{figure}
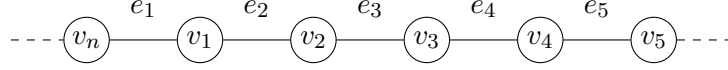

\begin{proof}
    The probability that an edge samples $0000$ is $1/16$, and the probability that it samples a value $\neq 0000$ is $15/16$. As a result, the bad event at node $i$ occurs with probability $2\cdot \frac{1}{16} \cdot \frac{15}{16} = 15/128$. The threshold for freezing is therefore $\sqrt{p} = \sqrt{15/128}\approx 5.48/16 \approx 43.8/128 \in (1/4,3/8)$.

    Let $B_i$ be the indicator random variable for the event at node $v_i$ occurring. Let us denote unsampled random variables at an edge by ${*}$, e.g., $e_i =01{*}{*}$ indicates that the first two random bits of $e_i$ were sampled (with values $0$ and $1$), while the last two are still unsampled. Consider what happens when sampling the variables of the first edge, $e_1$. Initially, $\Pr[B_i] = 15/128$. If the first sampled bit is $1$ ($e_1 = 1{*}{*}{*}$), then $e_1 \neq 0000$ is guaranteed, so $\Pr[B_1 \mid e_1 = 1{*}{*}{*}] = \Pr[e_2=0000 \mid e_1 = 1{*}{*}{*}] = 1/16$, which is $<\sqrt{p}$. Otherwise ($e_1 = 0{*}{*}{*}$), $e_1$ now has a probability of $1/8$ of being all $0$, so $\Pr[B_1 \mid e_1 = 0{*}{*}{*}] = 1/8*15/16 + 7/8*1/16 = 22/128 < \sqrt{p}$. Sampling a second $0$ yields a probability of $\Pr[B_1 \mid e_1 = 00{*}{*}] = 1/4*15/16 + 3/4*1/16 = 18/64 < \sqrt{p}$ of $B_1$ occurring, and a third $0$ gives a probability $\Pr[B_1 \mid e_1 = 000{*}] = 1/2*15/16 + 1/2*1/16 = 15/32$, exceeding the freezing threshold $\sqrt{p}$. Thus, the probabilities of the various outcomes of $e_1$ when sampling its random bits in the progressive sampling process are $\Pr[e_1 = 000{\froz}] = 1/8$ and $\Pr[e_1 \neq 0000] = 7/8$ (where $\neq 0000$ corresponds to sampling one of the $14$ $4$-bit strings with a $1$ in the first $3$ bits).

    Doing a similar analysis for $e_2$ shows that:
    \begin{itemize}
      \item $\Pr[e_2 = \froz\froz\froz\froz \wedge e_1 = 000\froz] = 1/8$
      \item $\Pr[e_2 = 000\froz \wedge e_1 \neq 0000] = 1/8\cdot7/8 = 7/64$,
      \item $\Pr[e_2 \neq 0000 \wedge e_1 \neq 0000] = 7/8\cdot7/8 = 49/64$.
    \end{itemize}

    More generally, for any $A \in \set{F_1, \neg F_1,F_1\vee\neg F_1}$, we have the following:
    \begin{itemize}
        \item $\Pr[e_{i+1} = \froz\froz\froz\froz \mid e_i = 000\froz \wedge A] = 1$,
        \item $\Pr[e_{i+1} = 000\froz \mid e_i = \froz\froz\froz\froz \wedge A] = 1/8$,
        \item $\Pr[e_{i+1} = 000\froz \mid e_i \neq 0000 \wedge A] = 1/8$,
        \item $\Pr[e_{i+1} \neq 0000 \mid e_i \neq 0000 \wedge A] = 7/8$,
        \item $\Pr[e_{i+1} \neq 0000 \mid e_i = \froz\froz\froz\froz \wedge A] = 7/8$,
    \end{itemize}
    For the remaining combinations of possible outcomes of the random variables of $e_{i+1}$ and $e_i$, e.g., for $\Pr[e_{i+1} =\froz\froz\froz\froz \mid e_i = \froz\froz\froz\froz \wedge A]$, all equal to $0$. 
    
    Let $X_{A,i}$ be defined as
    \begin{equation}
    X_{A,i}
    =
        \begin{pmatrix}
        \Pr[e_{i+1} = \froz\froz\froz\froz \wedge e_i = 000\froz \wedge A]
        \\
        \Pr[e_{i+1} = 000\froz \wedge e_i = \froz\froz\froz\froz \wedge A]
        \\
        \Pr[e_{i+1} \neq 0000 \wedge e_i = \froz\froz\froz\froz \wedge A]
        \\
        \Pr[e_{i+1} = 000\froz \wedge e_i \neq 0000 \wedge A]
        \\
        \Pr[e_{i+1} \neq 0000 \wedge e_i \neq 0000 \wedge A]
        \end{pmatrix}
    \end{equation}
{\color{red}

}
    We have
    \begin{equation}
    \label{eqn:recursion}
    X_{A,i+1}
    =
        \begin{pmatrix}
        0 & 1 & 0 & 1 & 0
        \\
        1/8 & 0 & 0 & 0 & 0 
        \\
        7/8 & 0 & 0 & 0 & 0 
        \\
        0 & 0 & 1/8 & 0 & 1/8 
        \\
        0 & 0 & 7/8 & 0 & 7/8
        \\
        \end{pmatrix}
        X_{A,i}
    \end{equation}
To see equation, e.g., the first term of \ref{eqn:recursion}, observe that 
\begin{align}
&\Pr[e_{i+2} = \froz\froz\froz\froz \mid e_{i+1} = 000\froz \wedge A] \cdot\Pr[e_{i+1} = 000\froz \wedge e_i = \froz\froz\froz\froz \wedge A] \\
&+  \Pr[e_{i+2} = \froz\froz\froz\froz \mid e_{i+1} = 000\froz \wedge A]\cdot \Pr[e_{i+1} = 000\froz \wedge e_i \neq 0000 \wedge A]  \\
&= \Pr[e_{i+2} = \froz\froz\froz\froz \mid e_{i+1} = 000\froz \wedge e_i = \froz\froz\froz\froz \wedge A] \cdot\Pr[e_{i+1} = 000\froz \wedge e_i = \froz\froz\froz\froz \wedge A] \\
&+  \Pr[e_{i+2} = \froz\froz\froz\froz \mid e_{i+1} = 000\froz \wedge e_i\neq 0000\wedge A]\cdot \Pr[e_{i+1} = 000\froz \wedge e_i \neq 0000 \wedge A] \\
&= \Pr[e_{i+2} = \froz\froz\froz\froz \wedge e_{i+1} = 000\froz \wedge e_{i} = \froz\froz\froz\froz \wedge A]
\\
&+\Pr[e_{i+2} = \froz\froz\froz\froz \wedge e_{i+1} = 000\froz \wedge e_{i} \neq 0000 \wedge A] \\
&= 
\Pr[e_{i+2} = \froz\froz\froz\froz \wedge e_{i+1} = 000 \froz \wedge A \wedge e_i = \froz\froz\froz\froz] \\
& +
\Pr[e_{i+2} = \froz\froz\froz\froz \wedge e_{i+1} = 000 \froz \wedge A \wedge e_i = 000\froz] \qquad (=0)\\
& +
\Pr[e_{i+2} = \froz\froz\froz\froz \wedge e_{i+1} = 000 \froz \wedge A \wedge e_i \neq 0000]\\
& =\Pr[e_{i+2} = \froz\froz\froz\froz \wedge e_{i+1} = 000\froz \wedge A]
\end{align}

    And
    \begin{equation}
    X_{F_1,1}
    =
        \begin{pmatrix}
        1/8\\
        0 \\
        0 \\
        7/64 \\
        0 \\
        \end{pmatrix}
    , \qquad
    X_{\neg F_1,1}
    =
        \begin{pmatrix}
        0\\
        0 \\
        0 \\
        0 \\
        49/64 \\
        \end{pmatrix}
    \end{equation}
By applying the recursion from \ref{eqn:recursion} $3$ times we obtain

$X_{F_1,4}=(1/32768)[896,64,448,784,5488]^T$ and $X_{\neq F_1,4}=(1/32768)[2744,392,2744,2401,16807]^T$. 
As $F_4$ holds if at least one of its incident edges $e_4$ or $e_5$ contains a frozen variable, we obtain 
\begin{align*}
\Pr(F_4\wedge F_1)& =(896+64+448+784)/32768=2192/32768~,
\\
\Pr(\neg F_4\wedge F_1) & = 5488/32768~,
\\
\Pr(F_4\wedge \neg F_1)  & = (2744+392+2744+2401)/32768=8281/32768\text{~, and }\\
\Pr(\neg F_4\wedge \neg F_1) & = 16807/32768~.
\end{align*} We obtain $\Pr(F_1)=(5488+2192)/32768=7680/32768$, and $\Pr(F_4)=(2192+8281)/32768=10473/32768$.

\begin{align}
\frac{2192}{32768} = \Pr(F_1\wedge F_4)
\;\neq\;
\Pr(F_1)\cdot \Pr(F_4)=\frac{7680}{32768}\cdot \frac{10473}{32768}\approx \frac{2455}{32768}
\end{align}

Consider the Jordan decomposition of the transition matrix $P$ 
\begin{align*}
P=\begin{pmatrix}
        0 & 1 & 0 & 1 & 0
        \\
        1/8 & 0 & 0 & 0 & 0 
        \\
        7/8 & 0 & 0 & 0 & 0 
        \\
        0 & 0 & 1/8 & 0 & 1/8 
        \\
        0 & 0 & 7/8 & 0 & 7/8
        \\
\end{pmatrix} = S \cdot J \cdot S^{-1}
\end{align*}

\begin{align*}
S
=
\frac{1}{392}
\begin{pmatrix}
        64 & 0 & 0 & -392 & 64
        \\
        -64 & -392 & -49 & 0 & 8 
        \\
        -448 & 0 & -343 & 0 & 56
        \\
        56 & 392 & 49 & 0 & 56
        \\
        392 & 0 & 343 & 392 & 392
        \\
\end{pmatrix}
,\quad
J=
\frac{1}{8}
\begin{pmatrix}
        -1 & 0 & 0 & 0 & 0
        \\
        0 & 0 & 0 & 0 & 0 
        \\
        0 & 0 & 0 & 8 & 0 
        \\
        0 & 0 & 0 & 0 & 0 
        \\
        0 & 0 & 0 & 0 & 8
        \\
\end{pmatrix}
,\\
S^{-1} =
\frac{1}{504}
\begin{pmatrix}
        2744 & -21952 & 2744 & -21952 & 2744
        \\
        0 & -504 & 72 & 0 & 0 
        \\
        -3528 & 28728 & -4104 & 28728 & -3528
        \\
        0 & -3528 & 504 & -3528 & 504
        \\
        343 & 343 & 343 & 343 & 343
        \\
\end{pmatrix}
\end{align*}
Note that $J^k$ only contains two non-zero entries, on the diagonal, for $k\geq 2$. This allows us to compute the following formulas for $\Pr[F_1 \wedge F_j]$ and $\Pr[\neg F_1 \wedge F_j]$ when $j \geq 3$:

\begin{align}
    X_{F_1,j}
    &
    = S \cdot J^{j-1} \cdot S^{-1} \cdot X_{F_1,1}
    =
    S \cdot 
    \frac{1}{504 \cdot 64 \cdot 8^{j-1}}
\begin{pmatrix}
        131712 \cdot (-1)^{j}
        \\
        0
        \\
        0
        \\
        0
        \\
        5145 \cdot 8^{j-1}
        \\
\end{pmatrix}
\\
&=
\frac{131712 \cdot (-1)^{j}}{392 \cdot 504 \cdot 64 \cdot 8^{j-1}}
\begin{pmatrix}
        64\\
        - 64\\
        - 448\\
        56\\
        392\\
\end{pmatrix}
+
\frac{5145 \cdot 8^{j-1}}{392 \cdot 504 \cdot 64 \cdot 8^{j-1}}
\begin{pmatrix}
        64
        \\
        8
        \\
        56
        \\
        56
        \\
        392
        \\
\end{pmatrix}
\end{align}

\begin{align*}
\Pr[F_1 \wedge F_j] 
& = \frac{(64-64-448+56)\cdot 131712}{392 \cdot 504 \cdot 64 \cdot 8^{j-1}} (-1)^{j} + \frac{(64+8+56+56)\cdot 5145}{392 \cdot 504 \cdot 64}
\\
& = \frac{115}{1536} + (-\frac{1}{8})^{j-1} \cdot \frac{49}{12}
\\
\Pr[F_1 \wedge \neg F_j] 
& = \frac{392 \cdot 131712}{392 \cdot 504 \cdot 64 \cdot 8^{j-1}} (-1)^{j} + \frac{392 \cdot 5145}{392 \cdot 504 \cdot 64} 
\\
& = \frac{245}{1536} - (-\frac{1}{8})^{j-1} \cdot \frac{49}{12}
\end{align*}

\begin{align}
    X_{\neg F_1,j}
    &
    = S \cdot J^{j-1} \cdot S^{-1} \cdot X_{\neg F_1,1}
    =
    S \cdot 
\frac{1}{504 \cdot 64 \cdot 8^{j-1}}
\begin{pmatrix}
        134456 \cdot (-1)^{j-1}
        \\
        0
        \\
        0
        \\
        0
        \\
        16807 \cdot 8^{j-1}
        \\
\end{pmatrix}
\\
&=
\frac{134456 \cdot (-1)^{j-1}}{392 \cdot 504 \cdot 64 \cdot 8^{j-1}}
\begin{pmatrix}
        64\\
        - 64\\
        - 448\\
        56\\
        392\\
\end{pmatrix}
+
\frac{16807 \cdot 8^{j-1}}{392 \cdot 504 \cdot 64 \cdot 8^{j-1}}
\begin{pmatrix}
        64
        \\
        8
        \\
        56
        \\
        56
        \\
        392
        \\
\end{pmatrix}
\end{align}

\begin{align*}
    \Pr[\neg F_1 \wedge F_j] 
&
= \frac{(64-64-448+56)\cdot 134456}{392 \cdot 504 \cdot 64 \cdot 8^{j-1}} (-1)^{j-1} + \frac{(64+8+56+56)\cdot 16807}{392 \cdot 504 \cdot 64}
\\
& = \frac{1127}{4608} - (-\frac{1}{8})^{j-1} \cdot \frac{2401}{576}
\\
\Pr[\neg F_1 \wedge \neg F_j]
& = \frac{392 \cdot 134456}{392 \cdot 504 \cdot 64 \cdot 8^{j-1}} (-1)^{j-1} + \frac{392 \cdot 16807}{392 \cdot 504 \cdot 64}
\\
& = \frac{2401}{4608} + \parens*{-\frac{1}{8}}^{j-1} \cdot \frac{2401}{576}
\end{align*}

Note how the ratios $\Pr[\neg F_1 \wedge F_j] / \Pr[\neg F_1 \wedge \neg F_j]$ and $\Pr[F_1 \wedge F_j] / \Pr[F_1 \wedge \neg F_j]$ converge to the same value, $115/245 = 1127/2401 = 23/49$.
\end{proof}

\section{Additional Pseudocode of Analyzed Algorithms}

\begin{algorithm}[ht]

    \caption{Pre-shattering of \cite{BEPS16_jacm} for Maximal Independent Set}

    \Input{A graph $G=(V,E)$ of maximum degree $\Delta$.}
    
    \Output{An independent set $I$ s.t.\ the graph induced by unsolved nodes is shattered.}

    \smallskip

    Initialize $I=\emptyset$, $B=\emptyset$

    \nonl Throughout, let $V_\sIB:= V \setminus(\hGamma(I) \cup B)$, and let $\deg_\sIB(v)$ (resp., $N_\sIB(v)$) be the degree (resp., neighborhood) of $v$ in the subgraph induced by $V_\sIB$.

    \For{each \emph{scale} k from $1$ to $\log \Delta +1$}{
        
        \For({(each iteration a \emph{Luby step})}){$T = c \log \Delta$ iterations}{
            If $v \in V_\sIB$, pick $b(v) \in \set{0,1}$ at random with $\Pr[b(v)=1] = 1/(\deg_\sIB(v)+1)$

            If $b(v)=1$ and $\forall u \in N_\sIB(v), b(u)=0$, then $v$ joins the independent set $I$.
        }

        If $v \in V_\sIB$ and $\deg_\sIB(v) > \Delta/2^k$, $v$ joins $B$.
    }
    
    Update $B \gets B \setminus \hGamma(I)$. Return $I$.
    
    \label[algorithm]{alg:BEPSMis}
\end{algorithm}

\begin{algorithm}

    \Input{A graph $G=(V,E)$, a matching $M \subseteq E(G)$ such that $G[U]$ has maximum degree $\hat{\Delta}$ where $U = V(G) \setminus \bigcup_{uv \in M}\set{u,v}$.}

    \Output{An extended matching $M$ s.t.\ the graph $G[U]$ induced by $U = V(G) \setminus \bigcup_{uv \in M}\set{u,v}$ is shattered}    

    \smallskip

    \For{$c \log \hat{\Delta}$ iterations}{
        Each unsolved node $v \in U$ makes a proposal an unsolved neighbor $\prop(v)$ chosen u.a.r.\ in $N(v) \cap U$.

        Each node $v' \in U$ that received a proposal accepts the one whose origin $v\in \prop^{-1}(v')$ has highest ID.
        
        Build a directed graph $H$ as follows: for each accepted proposal emitted by a node $v$ towards a node $v'$, let $v\rightarrow v'$ be an oriented edge in $H$.

        \For{each $v \in H$ s.t.\ $\deg_H(v)>0$}{

            \uIf{$\indeg_H(v) = 0$}{
                $b(v) \gets 0$
            }
            \uElseIf{$\outdeg_H(v) = 0$}{
                $b(v) \gets 1$
            }
            \Else({($\indeg_H(v) = \outdeg_H(v) = 1$)}){
                $b(v) \gets 0$ w.p.\ $1/2$, $1$ otherwise.
            }
        }
            
        Each oriented edge $v\rightarrow v'$ with $b(v)=0$ and $b(v')=1$ is added to the matching $M$.

        Update the set of unsolved nodes $U$: remove nodes new in the matching $M$ and nodes of degree $0$ in $G[U]$. 
    }
    
    \caption{Pre-shattering of \cite{BEPS16_jacm} for Maximal Matching, after degree reduction}
    \label{alg:mm-beps}
\end{algorithm}

\end{document}